\documentclass[12pt]{amsart}

\newtheorem{teo}{Theorem}[section]
\newtheorem{lem}[teo]{Lemma}

\newtheorem{cor}[teo]{Corollary}

\newtheorem{ex}[teo]{Example}

\def\supp{\mathop{\rm supp}}

\def\la{\lambda}

\def\f{\varphi}

\def\Wp{{\stackrel{\circ}{W}{}^1_2}}

\def\te{\theta}

\def\ss{\subset}

\def\a{\alpha}

\def\G{\Gamma}
\def\d{{\partial}}

\def\R{\mbox{\bf R}}

\def\C{\mbox{\bf C}}

\def\cal{\mathcal}

\title[Behavior  as $t\rightarrow \infty$ of solutions]
{Behavior as $t\rightarrow \infty$ of Solutions of
a problem in Mathematical Physics}

\author{S.\,D. Troitskaya}
\thanks{The research was partially supported
by a grant for supporting of scientific potential
of the higher school No.PNP.2.1.1.5031.}
\address{Institute of Content and Methods of Education of
the Russian Academy of Education (ISMO RAO),
Pogodinskaya ul., 8,
 Moscow 119435 Russia}
\email{troitsks@gmail.com}

\begin{document}

\begin{abstract}
A class of solutions, decaying  as $t\rightarrow \infty$,
 of a two-dimensional model problem on the
oscillations of an ideal rotating fluid in some
domains with angular points is constructed  explicitly. The existence of solutions whose
 $L_2$-norms decrease more rapidly than  any
negative power of $t$, is established.

\end{abstract}

\maketitle


\section*{Introduction}
In the paper, the first initial-boundary problem
for the Poincar\'e--Sobolev equation is considered,
\begin{equation}\label{1}
\frac {\d^2 }{\d t^2}\left( \frac {\d^2 p}{\d x^2}+
\frac {\d^2 p}{\d y^2} \right)+\frac {\d^2 p}{\d y^2}=0,\quad
(x,y;t)\in D\times (0,\infty),\,\,D\subset \R^2,
\end{equation}
\begin{equation}\label{2}
p|_{\d D\times (0,\infty)}=0,
\end{equation}
\begin{equation}\label{3}
p|_{t=0}=p_0(x,y),\quad p_t|_{t=0}=p_1(x,y),\quad
(p_0|_{\d D}=0, p_1|_{\d D}=0),
 \end{equation}
where $D$ is a bounded domain with piecewise smooth boundary and $D$
satisfies the cone condition.
Let functions  $p_0$ and $p_1$ belong  to the Sobolev space
$\Wp (D)$, which is a completion of the set
$ C_0^{\infty}(D)$ of infinitely  differentiable functions whose
supports are contained in $D$ with respect to the norm generated by
the inner product
 \begin{equation}\label{4}
(f,g)_{1}=\int\!\!\!\int\limits_D\left(
\nabla f,\nabla \overline {g}\right) dx\,dy.
 \end{equation}
By a \emph{generalized solution} of the problem
(in the sense of the theory of distributions, or the so-called generalized functions)
we mean the functions $p(x,y;t)$ with values in $\Wp (D)$ that are
   twice continuously differentiable with respect to $t$ and satisfy condition (\ref{3})
and the relation
\begin{equation}\label{5}
\int\limits_D(p_{ttx}\f_x+p_{tty}\f_y+p_y\f_y) dD=0 \quad
\forall t\in (0,\infty)
\end{equation}
keeps for any function  $\f \in C_0^{\infty}(D)$.
If the function $p$ is sufficiently smooth here, then
this generalized solution  is a classical solution.

This problem arises in hydrodynamics when describing
 small oscillations of a rotating ideal
 fluid in the following model two-dimensional  case, i.e., under
the assumption that the components of  velocity and
the pressure of the fluid depend only on time and on two spatial
variables, and the domain (the vessel) filled by the fluid is the
cylinder $ Q=\{(x,y,z)|(x,y)\in D, z\in \R\}$.
A linearized system of equations
describing the dynamics of a rotating fluid was first considered
by H.~Poincar\'e in \cite{Poin85}.
S.~L.~Sobolev in his well-known paper \cite{Sob54}
initiated the investigation of qualitative properties
of solutions of this system, and also of related
equations of type (\ref{1}) in diverse domains.
To study problem (\ref{1})--(\ref{3}),
it is natural to introduce the following operator $A$
acting on the Sobolev space $\Wp (D)$. The operator $A$ is
defined on any smooth function $h\in C_0^{\infty}(D)$ as
the solution of the problem
 \begin{equation}\label{6}
\Delta (A\,h)=\frac {\d^2 h}{\d y^2},\quad A\,h\in \Wp (D),
 \end{equation}
and then  is extended by continuity  to a bounded
operator on $\Wp (D)$ which is selfadjoint with respect to
the inner product  (\ref{4}). Using  the operator $A$, we can represent
 problem (\ref{1})--(\ref{3})  as an abstract Cauchy problem:
 \begin{equation}\label{7}
{\Psi}''=-A\Psi, \quad \Psi (0)=\Psi_0,\, {\Psi}'(0)=\Psi_1,
 \end{equation}
where $t\mapsto \Psi(x,y,t)$ stands for a function with values
in $\Wp (D)$, and  $\Psi_0,\, \Psi_1\in\Wp (D)$ (see, for example, \cite{Zel70}).
One of the main problems arising in the study
of  problem (\ref{7}) is that on the behavior of the solutions of (\ref{7}) as
$t\to \infty$, which  is  closely related to the structure of the spectrum
of the operator $A$, as is well known.

As is also well known,  for any domain $D$,  the spectrum of
the operator $A$ is the closed interval
$[0, 1]$; however,  the qualitative structure of the  spectrum  substantially  depends  on
the shape of the domain.
The investigation of the spectral properties of the operator $A$
was initiated in  \cite{Alex49}
and  continued by many other authors.
A  rich bibliography concerning the problem and its surrounding
can be found in \cite{ABIK70}--\cite{Fok94}.
The structure  of the spectrum of
the operator $A$ is known completely in two cases only, namely, if $\d D$ is an ellipse
or if it is a rectangle, provided that these figures are symmetric with respect to the axis $Oy$.
In these cases the spectrum of the operator $A$ is purely point,
i.e. $A$ admits a complete system of eigenfunctions.
Therefore, in this domains all solutions of  problem (\ref{7})
are almost periodic in time. At the same time
(see \cite{Alex60}), arbitrary small modifications of
the boundary $\d D$ can result in the occurrence of a continuous
spectrum for the operator $A$, whereas, for some deformations of the boundary of the domain,
a singular component of the spectrum can appear
(see \cite{Fok94}).

The present paper is devoted to the study of the problem (\ref{7})
for the case in which $D$ is a triangle,\footnote{The results of the
paper can be extended (with the corresponding modifications)
to a rather large class of domains with angular points
(see the remark at the end of the paper). However, in the main text,
 we deliberately avoid any generalizations to
maximally simplify the presentation of the material.}
 \begin{equation}\label{8}
D:=\{(x,y)\;|\;0<x<\frac 1{\alpha},\, 0<y<\alpha x\},\quad
(0<\alpha <+\infty ).
 \end{equation}
It follows from the papers \cite{Tro94,Tro99} of the author
that the spectrum of the operator $A$  is purely
continuous in this case. Therefore, the forthcoming investigation
of the behavior of solutions of  problem (\ref{7}), which are
no almost periodic functions of $t$, and, in particular, the existence problem
for solutions decaying in time,
is closely related to the
properties of differential solutions of the spectral
equation for the operator $A$.
In Section \ref{sec:pervaja} of the present paper we explicitly construct
a class of  differential solutions of this kind, and also prove
that the spectrum of the operator $A$ is absolutely continuous on the subspace which is the
closure of the linear span of these differential solutions.
This enables us to explicitly  construct  a class of
solutions of problem (\ref{7}) whose
$L_2$-norm decays as $t\to \infty$ (see Section \ref{sec:vtoraja}).
Moreover, it turns out that some solutions of
this class  tend to zero more rapidly than any negative
power of $t$, and, in the course of time,  the entire ``energy'' of these solutions
turns out to be  concentrated in an
arbitrarily small neighborhood of one of the vertices
of the triangle $D$.
On one hand, this result agrees with the possible behavior of solutions of the problem
under consideration in a neighborhood
of the angular points, in the form predicted in
\cite{Gre68}, whereas, on the
 other hand, it suggests the idea that, in the investigation
 of the motions of an actual rotating fluid in a vessel   with such a
 boundary, for large values of $t$,
one should consider the corresponding nonlinear systems.

Some results of the present paper were announced in
the papers
\cite{Tro10Postr}, \cite{Tro10Svoistva}.

\section{Absolutely continuous spectrum
of the operator $A$ for the domain $D$}\label{sec:pervaja}

1. Let $D$ be of the form (\ref{8}). Consider
the spectral problem for the operator $A$,
\begin{equation}\label{9}
A\,u=\lambda\,u.
\end{equation}
It can readily be seen that
$\lambda \in (0,1)$ is an eigenvalue of the operator $A$
if and only if the hyperbolic equation
\begin{equation}\label{10}
\frac {\d^2 u}{\d x^2}-\frac 1{a^2} \frac {\d^2 u}{\d
y^2}=0, \quad a^2=\frac{\la}{1-\la},
\end{equation}
has at least one nontrivial generalized solution
$u\in \Wp (D)$. It follows from the results of
\cite{Tro94,Tro99}  that there are no
solutions of this kind
for the domain $D$ under  consideration,
and the spectrum of the
operator $A$ is purely continuous,
$\sigma (A)=\sigma_c (A)$.
This means that, if
$E_{\lambda}$ is the spectral function   of the operator $A$ and
 ${h}$  is an arbitrary element of the space $\Wp (D)$, than the function
$(E_{\lambda}{h},{h})_1$ is continuous on $[0,1]$.
In this section we explicitly construct
a function $E_{\lambda}{h}$ for the elements ${h}$ in some
$A$-invariant subspace and prove that the
corresponding functions
$(E_{\lambda}{h},{h})_1$ are absolutely continuous.

As is well known,  for any ${h}$
and any interval $[\lambda_1,\lambda_2]\subset [0,\,1]$, the function
$U(\lambda):=E_{\lambda}{h}$
satisfies the equation
\begin{equation}\label{11}
A\left(U(\lambda_2)-U(\lambda_1)\right)=\int\limits_
{\lambda_1}^{\lambda_2}\lambda\,d U(\lambda),
\end{equation}
and, conversely, every function $U(\lambda):[0,\,1] \to \Wp (D)$,
 continuously depending  on the parameter $\lambda$,
and satisfying equation (\ref{11}) for any $\lambda_1,\,\lambda_2\in  [0,\,1]$
and the condition $U(0)=\bf{0}$,
where $\bf{0}$ stands for the zero element of the space $\Wp (D)$, is
necessarily of the  form
$U(\lambda)=E_{\lambda}{h}$, ${h}\in \Wp (D)$
(see, for example, \cite{SmirV}).
Every function $U(\lambda)$ of this kind is referred to as a \emph{differential
solution} of  equation (\ref{9}).
In the next subsection, we construct a certain class of
differential solutions of the spectral equation for
the operator $A$.

2. Consider the hyperbolic equation (\ref{10}) in the domain
$D$. We refer to the characteristics
of the form $x= a\,y+c'$ of (\ref{10}) as the characteristics of the
first family  and those of the form $x= -a\,y+c''$
as the characteristics of the
second family. Choose  an arbitrary value $\lambda \in (0,(1+\alpha^2)^{-1})$.
Introduce  the following polygonal lines by
introducing the following reflection
law of the rays of the characteristic directions
at the boundary.
Suppose that a ray issuing from the point $B(1/{\alpha},1)$
in the direction of the first family of characteristics
 inside the domain $D$ meets the  boundary $\d D$ and goes inside $D$, after
the reflection at the boundary, already in the
direction of the  characteristics of the other family, after which, this ray meets the
 boundary $\d D$ again, makes a reflection at $\d D$, and goes along
 a characteristics of the first family, and so on.
 We also assume that the ray  issued from the point $A(1/{\alpha},0)$
 in the direction of the second family of characteristics inside the domain $D$
 admits a similar behavior under the reflection at $\d D$.
Note that, for the values of
$\lambda \in (0,(1+\alpha^2)^{-1})$
under consideration, both the rays ``hide'' into the
 the angle with the vertex $O(0,0)$. The polygonal lines that are
trajectories of these rays divide $D$ into
infinitely many triangles
$T_{\lambda}^i$ and parallelograms $P_{\lambda}^i$.
Denote the vertices of these polygonal lines
belonging to the segment  $[OB]$
by $B_{\lambda}^i$, those on the segment $[OA]$  by
$A_{\lambda}^i$, and the interior  points of their intersections by $C_{\lambda}^i$, $i=1,2,..$,
and index them in the ascending order in the direction from the segment
$[AB]$ to the point $O(0,0)$.

Let $\te_1  \in L_2(0,1)$ be an arbitrary function.

As is well known (see \cite{Tro94}), every generalized solution
$u_0(x,y,\lambda )$ of (\ref{10}) on
 the characteristic triangle  $ABC_{\lambda}^1$ such that
$u_0(x,y,\lambda )\in W^1_2(ABC_{\lambda}^1)$
coincides on this triangle almost everywhere  with
a function continuous on $\overline{ABC_{\lambda}^1}$
and such that the generalized derivative of this function
with respect to the variable $x$ (treated
   in the sense of distribution theory)  has a trace on the segment $[AB]$, and this trace
belongs to the space $L_2(0,1)$.
Therefore, one can pose the following problem:
among all solutions of (\ref{10}), find solutions that satisfy the
conditions
$$
u_0|_{AB}=0,\quad \left.\frac {\d u_0}{\d x}\right|_{AB}=\te_1.
$$
It follows from \cite{Tro98}  that such a solution $u_0$
exists and is unique in the triangle  $ABC_{\lambda}^1$.
Moreover, the trace of this solution
on the segment  $[BC_{\lambda}^1]$ belongs to
$W^1_2(BC_{\lambda}^1)$
and satisfies the following bound:
\begin{equation}\label{11m}
\left\|u_{0}|_{BC_{\lambda}^1}\right\|_{ W^1_2(BC_{\lambda}^1)}\leq
C_{0,\varepsilon}\|\te_1\|_{L_2},\quad
\end{equation}
where the constant $C_{0,\varepsilon} $ does not depend on $\te_1 $.

Further, using the function $u_0$, we can construct a generalized solution
$u_1(x,y,\lambda )\in W^1_2(BC_{\lambda}^1B_{\lambda}^1)$
of  (\ref{10}) in the triangle $BC_{\lambda}^1 B_{\lambda}^1$  such that
$$
u_1|_{BC_{\lambda}^1}=u_0|_{BC_{\lambda}^1},
\quad  u_1|_{BB_{\lambda}^1}=0.
$$
The problem of finding such a solution is a generalization
of the classical Darboux problem in which a solution of
a hyperbolic equation is uniquely determined by the
values of this solution on two curves issuing from a given point.
In the present case,  one of the curves
is a characteristic of the equation.
In  \cite{Tro99}, this problem is studied both for the regular and for the generalized
solutions. It follows from the results of  \cite{Tro99} that the desired
function $u_1$  exists and is unique and, moreover,
the trace of $u_1$ on the segment  $[C_{\lambda}^1 B_{\lambda}^1]$
belongs to the space
$W^1_2( C_{\lambda}^1 B_{\lambda}^1)$
and satisfies a bound similar  to (\ref{11m}).

After this, using the functions $u_0$ and $u_1$,
we construct a generalized solution
$u_2(x,y,\lambda)\in W^1_2(AB_{\lambda}^1 A_{\lambda}^2)$
of (\ref{10}) on the triangle
$AB_{\lambda}^1A_{\lambda}^2$  such that
$$
u_2|_{AC_{\lambda}^1}=u_0|_{AC_{\lambda}^1},\quad
u_2|_{C_{\lambda}^1B_{\lambda}^1}=
u_1|_{C_{\lambda}^1B_{\lambda}^1},\quad  u_2|_{AA_{\lambda}^2}=0.
$$
This process can be continued. Let
  \begin{equation} \label{12}
u(x,y;\te_1;\lambda):=\left\{
\begin{array}{ll}
u_0(x,y,\lambda ) &\mbox{ for } (x,y)\in \Delta ABC_{\lambda}^1,\\
u_1(x,y,\lambda ) &\mbox{ for } (x,y)\in \Delta BC_{\lambda}^1B_{\lambda}^1,\\
u_2(x,y,\lambda ) &\mbox{ for } (x,y)\in \Delta AB_{\lambda}^1A_{\lambda}^2, \\
\dots       &   \dots
\end{array}
\right.
 \end{equation}
It can readily be proved that, for any $\te_1 \in L_2(0, 1)$
and for any $\lambda \in (0,(1+\alpha^2)^{-1})$,
the function $u(x,y;\te_1;\lambda)$ thus constructed has
 the following properties:

a) the function $u(x,y;\te_1;\lambda)$ belongs to $L_2(D)$
and is a generalized solution of  (\ref{10});

b) for any $0<\varepsilon <1$, the function
$u(x,y;\te_1;\lambda)$ belongs to the space
$W^1_2(D\cap \{x>\varepsilon\})$;

c) for any smooth curve
$\vec{r}(t):=\{x(t),\,y(t),\,\lambda(t)\},\,t\in [t_0,\,t_1]$ lying
in the domain
$D\times (0,\,(1+\alpha^2)^{-1})$, the function
$u(x(t),y(t);\te_1;\lambda(t))$ is absolutely continuous on $[t_0,t_1]$;

d) for almost all $x\in (0,\,1/\alpha)$, the derivatives
$ u'_x(x,\alpha x)$, $u'_y(x,\alpha x)$, $u'_x(x,0)$, and $u'_y(x,0)$
are well defined and, for any $0<\varepsilon <\frac 1{\alpha} $,
there is a constant $C_{\varepsilon}$ independent
of $\te_1$ and such that
$$
\left\|u'_x(x,\alpha x)\right\|_{L_2(\varepsilon,\,\alpha)}\leq
C_{\varepsilon}\|\te_1\|_{L_2},\quad
\left\|u'_y(x,\alpha x)\right\|_{L_2(\varepsilon,\,\alpha)}\leq
C_{\varepsilon}\|\te_1\|_{L_2},\quad
$$
$$
\left\|u'_x(x,0)\right\|_{L_2(\varepsilon,\,\alpha)}\leq
C_{\varepsilon}\|\te_1\|_{L_2},\quad
\left\|u'_x(x,0)\right\|_{L_2(\varepsilon,\,\alpha)}\leq
C_{\varepsilon}\|\te_1\|_{L_2}.
$$

Let
$0<\lambda_*< \lambda_{**}<(1+\alpha^2)^{-1}$,
and let
$\sigma(\mu)\in C^1[\lambda_{*},\lambda_{**}]$ be a function. Introduce
 the  function:
 \begin{equation} \label{13}
U(x,y;\te_1;\sigma;\lambda):=\left\{
\begin{array}{ll}
\bf{0},       &\mbox{ for } \lambda\leq \lambda_*,\\
\int\limits_{\lambda_*}^{\lambda}\sigma (\mu)u(x,y;\te_1;\mu)d\mu,
&\mbox{ for }  \lambda_*<\lambda\leq\lambda_{**} \\
\int\limits_{\lambda_*}^{\lambda_{**}}\sigma (\mu)u(x,y;\te_1;\mu)d\mu,
&\mbox{ for }  \lambda_{**}<\lambda , \\
\end{array}
\right.
 \end{equation}
where $\bf{0}$ stands for the zero element of the space $\Wp (D)$.

\begin{teo}\label{teo:odinn}
The function
$U(x,y;\te_1;\sigma;\lambda)$
is a differential solution of  $(\ref{9})$.
 \end{teo}
 We present the proof of the theorem under
several headings.

\begin{lem}\label{lem:odin}
For any  $\lambda_1$, $\lambda_2$ such that
$\lambda_*\leq \lambda_1<\lambda_2\leq \lambda_{**}$,
the function
\begin{equation}\label{V}
V(x,y;\te_1;\sigma):=\int\limits_{\lambda_1}^{\lambda_2}
\sigma(\mu)
u(x,y;\te_1;\mu) d\mu
\end{equation}
belongs to the space $\Wp (D)$ and satisfies the bound
\begin{equation}\label{ocV}
\|V(x,y;\te_1;\sigma)\|_1\leq M\,(\lambda_2-\lambda_1)\|\te_1\|_{L_2},
\end{equation}
where $M$ does not depend on  $\lambda_1,\,\lambda_2,\,\te_1$.
 \end{lem}

\begin{proof}[Proof of the lemma]
Let us prove first that $V(x,y;\te_1;\sigma)\in\Wp (D)$.
By property b), to this end, it  suffices
to show that
$$
\left\| V'_x\right\|_{L_2({\cal D}_2)}<\infty,\quad
\left\| V'_y\right\|_{L_2({\cal D}_2)}<\infty,
$$
where
${\cal D}_2:=D\cap \{a_2y>x+a_2-\frac 1{\alpha}\}$ and $
a_2=\sqrt{\frac {\lambda_2}{1-\lambda_2}}$. For example,
let us prove the first inequality (the other can be proved in a similar way).
For  $(x,y)\in {\cal D}_2$,  for all
$\mu \in [\lambda_1,\lambda_2]$,
 by the Riemann formula (see \cite{Tro94}) we have
\begin{equation}\label{13,5}
u(x,y;\te_1 ;\mu)=\frac {\sqrt{\mu}}{2\sqrt{1-\mu}}
 \int\limits_{P(x,y;\mu)}^{Q(x,y;\mu)}
\f(x',\mu)dx',
\end{equation}
where
\begin{equation}\label{phi}
\f(x,\mu):=\left.\left(\alpha u'_x
+\frac {1-\mu}{\mu} u'_y\right)\right|_{y=\alpha x},
\end{equation}
  and $ P(x,y;\mu)$ and $Q(x,y;\mu)$ are abscissaes  of
  the right and left angles of the characteristic triangle
  with the vertex at the point $(x,y)$ (this triangle  corresponds to the value
  $\lambda=\mu$ and leans on $OB$). It can readily be established
  that
\begin{equation}\label{P,Q}
P(x,y;\mu)=\frac {\alpha x-y}{2\alpha}l+ \frac {\alpha x+y}{2\alpha},
\quad Q(x,y;\mu)=\frac {\alpha x+y}{2\alpha}+
\frac {\alpha x-y}{2\alpha l},
\end{equation}
where
\begin{equation}\label{l(mu)}
l=l(\mu):=\frac{\sqrt{1-\mu}+\alpha\sqrt{\mu}}
{\sqrt{1-\mu}-\alpha\sqrt{\mu}}
\end{equation}
is a function strictly monotone increasing on the interval
$(0,(1+\alpha^2)^{-1})$ and taking the values in the
interval $(1,+\infty)$, and therefore
\begin{equation}
l(\mu)>1\quad\mbox{ for any } \mu\in [\lambda_1,\,\lambda_2].
\end{equation}

Thus, if $(x,y)\in {\cal D}_2$, then
$$
V(x,y;\te_1;\sigma)=\int\limits_{\lambda_1}^{\lambda_2}
 \frac {\sigma(\mu)\sqrt{\mu}}{2\sqrt{1-\mu}}
\left\{
\int\limits_{P(x,y;\mu)}^{Q(x,y;\mu)}
\f (\widetilde x,\mu)\,d\widetilde x\right\}d\mu,
$$
and therefore
$$
\left\| V'_x\right\|_{L_2
({\cal D}_2)}^2
=\int\!\!\!\int\limits_{{\cal D}_2}\left|\,
\int\limits_{\lambda_1}^{\lambda_2}
\frac {\sigma(\mu)\sqrt{\mu}}{2\sqrt{1-\mu}}
\left\{  Q'_x\,\f(Q,\mu) -
 P'_x\,\f(P,\mu)
\right\}d\mu\right|^2dxdy\leq
$$
\begin{equation}\label{15}
\leq
2\int\!\!\!\int\limits_{{\cal D}_2}\left|\,
\int\limits_{\lambda_1}^{\lambda_2}
\frac {\sigma(\mu)\sqrt{\mu}}{2\sqrt{1-\mu}}
  P'_x\,\f(P,\mu)
d\mu\right|^2dxdy
+2\int\!\!\!\int\limits_{{\cal D}_2}\left|\,
\int\limits_{\lambda_1}^{\lambda_2}
\frac {\sigma(\mu)\sqrt{\mu}}{2\sqrt{1-\mu}}
 Q'_x\,\f(Q,\mu)
d\mu\right|^2dxdy.
\end{equation}
Thus, our problem is reduced to the existence problem for these integrals, for any
function
$\te_1 \in L_2(0, 1)$ and any
$\sigma(\mu)\in C^1[\lambda_{*},\lambda_{**}]$.

Note that the set of all piecewise constant
functions defined on the interval $(0,\, 1)$ and such that the ends of the
intervals on which the functions are constant are rationals is dense in $L_2(0, 1)$,
and therefore, for $\te_1 (x)$, we take  the  function
 \begin{equation} \label{17}
\te_{1,n}(x)=\left\{
\begin{array}{ll}
\dots       & \dots \\
c_i, & \frac{i-1}{n}<x<\frac {i}{n},\\
\dots       &   \dots
\end{array}
\right.
 \end{equation}
where $c_i\in \C$ are some numbers, $i=1,2,..,n $. In this case, the function
 $\f$ can be presented explicitly. Namely, let
\begin{equation}\label{mu}
\mu=\frac{(l-1)^2}{\alpha^2(l+1)^2+(l-1)^2},
\end{equation}
(it can readily be seen that  $\mu$ is the  function inverse to
$l(\mu)$, i.e.
$\mu(l(t))\equiv t$). Denote
$$
\widetilde{\f}(x,l):=\f(x,\mu(l)).
$$
Introduce the functions
\begin{equation}\label{x(k,j)}
 x_{k,j}:=\frac 1{\alpha l^k}\cdot
\frac {n+j}{2n}+\frac 1{\alpha l^{k+1}}\cdot\frac {n-j}{2n},\,\,
  k=0,1,2,...,\,\,j=0,\pm 1,...,\pm (n-1).
\end{equation}
Then
\begin{equation} \label{Phi}
\Phi(x,l):={\widetilde{\f}(x,l)}\frac {l-1}{2\alpha l}
=\left\{
\begin{array}{ll}
\dots       &   \dots \\
c_{n}l^k,       & x_{k,n-1}(l)<x<x_{k,n}(l), \\
\dots       &   \dots \\
c_il^k,       & x_{k,i-1}(l)<x<x_{k,i}(l), \\
\dots       &   \dots \\
c_1l^k,       & x_{k,0}(l)<x<x_{k,1}(l),\\
-c_1l^k,       & x_{k,-1}(l)<x<x_{k,0}(l),\\
-c_il^k,       & x_{k,-i}(l)<x<x_{k,1-i}(l), \\
\dots       &   \dots \\
-c_nl^k,       & x_{k,-n}(l)<x<x_{k,1-n}(l),\\
\dots       &   \dots\\
\end{array}
\right.
 \end{equation}
where $k=0,1,2,...,\,\,i=1,2,..,n$.

For example, consider the first integral on
the right-hand side of (\ref{15}).
It is clear that the existence of the integral is equivalent
to the convergence of the series
\begin{equation}\label{16}
\sum\limits_{k=0}^{+\infty}
\int\!\!\!\!\int\limits_{R_k}
\left|\int\limits_{\lambda_1}^{\lambda_2}
\frac {\sigma(\mu)\sqrt{\mu}}{2\sqrt{1-\mu}}
  P'_x\,\f(P,\mu)d\mu\right|^2dxdy,
\end{equation}
where
$$
R_k:=\left
\{(x,y)\in {\cal D}_2\,\left|
\,\frac 1{\alpha l_1^{k+1}}<x<\frac 1{\alpha l_1^{k}}\right.\right\},
\quad l_1:=l(\lambda_1)>1.
$$

Introduce the following notation:
$$
\widetilde{\sigma}(l):=
\frac {\sigma(\mu(l))\sqrt{\mu(l)}}{2\sqrt{1-\mu(l)}}
\,\mu'(l)\,\frac {2\alpha l}{l-1},\qquad
\widetilde{P_x}(x,y;l):= P'_x(x,y;\mu(l)).
$$
In this case,
$$
\int\limits_{\lambda_1}^{\lambda_2}
\frac {\sigma(\mu)\sqrt{\mu}}{2\sqrt{1-\mu}}
  P'_x(x,y;\mu)\,\f(P(x,y;\mu),\mu)d\mu=
$$
\begin{equation}\label{19}
=\int\limits_{l_1}^{l_2}\widetilde{\sigma}(l)
\widetilde{P_x}(x,y;l)\,\Phi(\widetilde{P}(x,y;l),l)
\,dl,\quad l_i=l(\lambda_i),\,\,i=1,2.
\end{equation}

For any chosen point $(x,y)\in R_k$, denote
by $l_{k,k',j}$
the values of
 $l$ corresponding to the intersection points  of the curve
$x=\widetilde{P}(x,y;l)$ with the curves $x_{k',j}(l)$ in the plane
$Olx$. Let $N(x,y;k)$ be the number of the values
$l_{k,k',j}$ belonging to the interval
$[l_1,\,l_2]$. Assume that $k'$ takes here the values  $\{k'_m, k'_m+1,...,k'_M\}$.
It can readily be  seen that, for sufficiently large values of
$k$, there is a constant $C_0$ (independent of $(x,y),\,k,\,l_1,\,l_2$ and
 depending on $\lambda_*,\,\lambda_{**}$ only) such that
$$
k\log_{l_2}l_1 -C_0 < k'_m\leq k'\leq k'_M\leq k+1.
$$
Therefore, for sufficiently large values of $k$,
$$
N(x,y;k)\leq k+1- k\log_{l_2}l_1 +C_0 \,\,(  < \gamma k),
$$
where the number $0<\gamma <1$ does not depend
on $(x,y),\,k,\,l_1,\,l_2$.

In this case, for any point $(x,y)\in R_k$, we have
$$
\int\limits_{l_1}^{l_2}{\cal P}
\,dl=
\int\limits_{l_1}^{l_{k,k'_M,n}} {\cal P}dl +
\int\limits_{l_{k,k'_M,n}}^{l_{k,k'_M-1,n}} {\cal P}dl +...+
\int\limits_{l_{k,k'_m,n}}^{l_2} {\cal P}dl,
$$
where  ${\cal P}={\cal P}(x,y,l)$
stands for the integrand in (\ref{19}).
Here, for any $k'=k'_m+1,...,k'_M$,
$$
\int\limits_{l_{k,k',n}}^{l_{k,k'-1,n}} {\cal P}dl =
\sum\limits_{i=1}^{n}
c_i\left(
\int\limits_{l_{k,k'-1,i-1}}^{l_{k,k'-1,i}}
\widetilde{\sigma}(l)
\widetilde{P_x}(x,y;l)l^{k'-1}dl-
\int\limits_{l_{k,k'-1,-i}}^{l_{k,k'-1,1-i}}
\widetilde{\sigma}(l)
\widetilde{P_x}(x,y;l)l^{k'-1}dl\right),
$$
and we can change the order of  summation, because all sums contain
finitely many summands for any chosen $k$ and $n$, and therefore
$$
\int\limits_{l_1}^{l_2}
{\cal P}\,dl=
\int\limits_{l_1}^{l_{k,k'_M,n}} {\cal P}dl +
\int\limits_{l_{k,k'_m,n}}^{l_2} {\cal P}dl+
$$
\begin{equation}\label{19,5}
+\sum\limits_{i=1}^{n}
c_i  \sum\limits_{k'=k'_m+1}^{k'_M}\left(
\int\limits_{l_{k,k'-1,i-1}}^{l_{k,k'-1,i}}
\widetilde{\sigma}(l)
\widetilde{P_x}(x,y;l)l^{k'-1}dl-
\int\limits_{l_{k,k'-1,-i}}^{l_{k,k'-1,1-i}}
\widetilde{\sigma}(l)
\widetilde{P_x}(x,y;l)l^{k'-1}dl\right).
\end{equation}
Thus, for any point $(x,y)\in R_k$,
$$
\left|\int\limits_{l_1}^{l_2} {\cal P}
\,dl\right|^2\leq
3 \left|\int\limits_{l_1}^{l_{k,k'_M,n}}{\cal P}dl\right|^2 +
3\left|\,\int\limits_{l_{k,k'_m,n}}^{l_2} {\cal P}dl\right|^2+
$$
$$
+3\left|\sum\limits_{i=1}^{n}
c_i  \sum\limits_{k'=k'_m+1}^{k'_M}\left(\,\,
\int\limits_{l_{k,k'-1,i-1}}^{l_{k,k'-1,i}}
\widetilde{\sigma}(l)
\widetilde{P_x}(x,y;l)l^{k'-1}dl-
\int\limits_{l_{k,k'-1,-i}}^{l_{k,k'-1,1-i}}
\widetilde{\sigma}(l)
\widetilde{P_x}(x,y;l)l^{k'-1}dl\right)\right|^2 .
$$
Let us estimate the last summand.
Denote $t(x,y;l):=
\widetilde{\sigma}(l) \widetilde{P_x}(x,y;l)$.
Since  $t(x,y;l)\in C^1[l_{*},l_{**}]$
for any point $(x,y)\in R_k$,  where
$l_{*}:=l(\lambda_*),\,l_{**}:=l(\lambda_{**})$, it follows that
$$
\int\limits_{l_{k,k'-1,i-1}}^{l_{k,k'-1,i}}
\widetilde{\sigma}(l)
\widetilde{P_x}(x,y;l)l^{k'-1}dl-
\int\limits_{l_{k,k'-1,-i}}^{l_{k,k'-1,1-i}}
\widetilde{\sigma}(l)
\widetilde{P_x}(x,y;l)l^{k'-1}dl=
$$
$$
=\frac{1}{k'}\left(\left.
t(x,y;l)l^{k'}\right|^{l_{k,k'-1,i}}_{l_{k,k'-1,i-1}}-
\left.t(x,y;l)l^{k'}\right|^{l_{k,k'-1,1-i}}_{l_{k,k'-1,-i}}
\right)-
$$
\begin{equation}\label{20}
-\int\limits_{l_{k,k'-1,i-1}}^{l_{k,k'-1,i}}
\frac{l^{k'}}{k'}t'_{l}(x,y;l)dl+
\int\limits_{l_{k,k'-1,-i}}^{l_{k,k'-1,1-i}}
\frac{l^{k'}}{k'}t'_{l}(x,y;l)dl.
\end{equation}
Note that, by the choice of $l_{k,k',j}$,
\begin{equation}\label{21}
 l_{k,k'-1,i}^{k'}t(x,y; l_{k,k'-1,i})=
\frac{\frac {i+n}{2n}l_{k,k'-1,i}+\frac {n-i}{2n}}{
\widetilde{P}(x,y;l_{k,k'-1,i})}\;t(x,y; l_{k,k'-1,i}).
\end{equation}
Denote
$$
G_i(z):=
\frac{\frac {i+n}{2n}z+\frac {n-i}{2n}}{
\widetilde{P}(x,y;z)}\;t(x,y;z),
 \quad i=0,\pm 1,\pm 2,...\pm n.$$
Then
$$
\left.t(x,y;l)l^{k'}\right|^{l_{k,k'-1,i}}_{l_{k,k'-1,i-1}}=
G_i(l_{k,k'-1,i})-G_i(l_{k,k'-1,i-1})+
\frac{\frac {1}{2n}(1-l_{k,k'-1,i-1})}{
\widetilde{P}(x,y;l_{k,k'-1,i-1})}t(x,y; l_{k,k'-1,i-1}).
$$
Similarly,
$$
\left.t(x,y;l)l^{k'}\right|^{l_{k,k'-1,1-i}}_{l_{k,k'-1,-i}}=
G_{-i}(l_{k,k'-1,1-i})-G_{-i}(l_{k,k'-1,-i})+
\frac{\frac {1}{2n}(1-l_{k,k'-1,1-i})}{
\widetilde{P}(x,y;l_{k,k'-1,1-i})}t(x,y; l_{k,k'-1,1-i}).
$$
Therefore, writing
$$
K(z):=\frac {z-1}{\widetilde{P}(x,y;z)}t(x,y;z),
$$
we obtain
$$
\left.
t(x,y;l)l^{k'}\right|^{l_{k,k'-1,i}}_{l_{k,k'-1,i-1}}-
\left.t(x,y;l)l^{k'}\right|^{l_{k,k'-1,1-i}}_{l_{k,k'-1,-i}}=
$$
$$
=\left(G_i(l_{k,k'-1,i})-G_i(l_{k,k'-1,i-1})\right)-
\left(G_{-i}(l_{k,k'-1,1-i})-G_{-i}(l_{k,k'-1,-i})\right)+
$$
$$
+\frac 1{2n}\left(K(l_{k,k'-1,i-1})-K(l_{k,k'-1,1-i})\right).
$$
It can readily be proved that there are
constants $C_2,\,C_3$ independent of
$l_1,\,l_2,\,n,\,k$ and such that
$$
\left|G'_{i}(z)\right|\leq C_2 \,l_1^k,\quad
\left|K'(z)\right|\leq C_3 \,l_1^k
$$
for any point  $(x,y)\in R_k$, for any
 $i=0,\pm 1,\pm 2,...\pm n$, and for any
$z\in [l_{*},\,l_{**}]$. Hence,
$$
\frac 1{k'}\Biggl|\left.
t(x,y;l)\;l^{k'}\right|^{l_{k,k'-1,i}}_{l_{k,k'-1,i-1}}-
\left.t(x,y;l)\;l^{k'}\right|^{l_{k,k'-1,1-i}}_{l_{k,k'-1,-i}}
\Biggr|\leq
$$
$$
\leq
\frac {l_1^k}{k'}\left(C_2(l_{k,k'-1,i}-l_{k,k'-1,i-1})+
C_2(l_{k,k'-1,1-i}-l_{k,k'-1,-i})+
\frac 1{2n} C_3  (l_{k,k'-1,i-1}-l_{k,k'-1,1-i})\right).
$$
Further, by the relation
$$
\frac{x_{k'-1,i}(l_{k,k'-1,i-1})-x_{k'-1,i-1}(l_{k,k'-1,i-1})}
{x_{k'-1,n}(l_{k,k'-1,i-1})-x_{k'-1,-n}(l_{k,k'-1,i-1})}=
\frac 1{2n},
$$
there exists a constant
$C_{4}$ independent of $n,\,k,\,(x,y),\, k',\, i$ and such that
$$
\left(l_{k,k'-1,i}-l_{k,k'-1,i-1}\right)  \leq
\frac {C_{4}} {2n} (l_{k,k'-1,n}-l_{k,k'-1,-n}),
$$
and therefore
$$
\frac 1{k'}\Biggl|\left.
t(x,y;l)\;l^{k'}\right|^{l_{k,k'-1,i}}_{l_{k,k'-1,i-1}}-
\left.t(x,y;l)\;l^{k'}\right|^{l_{k,k'-1,1-i}}_{l_{k,k'-1,-i}}
\Biggr|\leq
$$
$$
\leq\frac {C_5\,l_1^k}{2nk'}\left(
  l_{k,k'-1,n}-l_{k,k'-1,-n}\right),
$$
where $C_{5}$ does not depend on $l_1,\,l_2,\,n,\,k,\,(x,y),
\, k',\, i$.

Let us now obtain  similar bounds for the integrals
in the right-hand side of  (\ref{20}). Consider the first integral.
Since the functions in the integrand are continuous, it follows that
$$
\int\limits_{l_{k,k'-1,i-1}}^{l_{k,k'-1,i}}
\frac{l^{k'}}{k'}\;t'_{l}(x,y;l)\,dl=
\frac{t'_{l}(x,y;\eta_{k,k',i})\;\eta_{k,k',i}}{k'}
\int\limits_{l_{k,k'-1,i}}^{l_{k,k'-1,i-1}}l^{k'-1}\,dl=
$$
$$
=\frac{t'_{l}(x,y;\eta_{k,k',i})\;\eta_{k,k',i}}{(k')^2}
\left(l_{k,k'-1,i}^{k'}-   l_{k,k'-1,i-1}^{k'}\right),
$$
where $\eta_{k,k',i}\in [l_{k,k'-1,i-1}, l_{k,k'-1,i}]$. We have
(see (\ref{21})):
$$
l_{k,k'-1,i}^{k'}-   l_{k,k'-1,i-1}^{k'}=
F_i(l_{k,k'-1,i})-F_i(l_{k,k'-1,i-1})+
\frac{\frac {1}{2n}(1-l_{k,k'-1,i-1})}{
\widetilde{P}(x,y;l_{k,k'-1,i-1})},
$$
where
$$
F_i(z):=
\frac{\frac {i+n}{2n}z+\frac {n-i}{2n}}{
\widetilde{P}(x,y;z)}, \quad i=0,\pm 1,\pm 2,...\pm n.
$$
Here there are constants
$C_6,\,\,C_7$ independent of $l_1,\,l_2,\,n,\,k,\,k',\,i$ and such  that
the following bounds hold
for any points $(x,y)\in R_k$,
for any  $i=0,\pm 1,\pm 2,...\pm n$, and
for any $z\in [l_{*},\,l_{**}]$:
$$
\left|F'_{i}(z)\right|\leq C_6 \;l_1^k,\quad
\left|\frac{(1-l_{k,k'-1,i-1})}{
\widetilde{P}(x,y;l_{k,k'-1,i-1})}\right|\leq C_7\; l_1^k,
$$
and therefore
$$
\left|\,\,\int\limits_{l_{k,k'-1,i-1}}^{l_{k,k'-1,i}}
\frac{l^{k'}}{k'}\;t'_{l}(x,y;l)\,dl\right|\leq
\frac {C_8 \,l_1^k}{(k')^2}
\left(
l_{k,k'-1,i}-l_{k,k'-1,i-1}\right)  +
\frac {C_9 \,l_1^k}{2n\,(k')^2}\leq
$$
$$
\leq
\frac {C_{10} \,l_1^k}{2n\,(k')^2} \left(
l_{k,k'-1,n}-l_{k,k'-1,-n}\right)  +
\frac {C_9 \,l_1^k}{2n\,(k')^2},
$$
where  $C_8,\,\,C_9,\,\,C_{10}$ are constants independent of  $l_1,\,l_2,\,n,\,k,\,(x,y),\,k',\,i$.
It is clear that  the second integral summand in  (\ref{20})
can be estimated in a similar way, and therefore
$$
\left|\,
\int\limits_{l_{k,k'-1,i-1}}^{l_{k,k'-1,i}}
\widetilde{\sigma}(l)\;
\widetilde{P_x}(x,y;l)\;l^{k'-1}\,dl-
\int\limits_{l_{k,k'-1,-i}}^{l_{k,k'-1,1-i}}
\widetilde{\sigma}(l)\;
\widetilde{P_x}(x,y;l)\;l^{k'-1}\,dl\right|\leq
$$
$$
\leq\frac {C_{11} \,l_1^k}{n\,k'}
\left(
l_{k,k'-1,n}-l_{k,k'-1,-n}\right)  +
\frac {C_{12} \,l_1^k}{2n\,(k')^2},
$$
where the constants $C_{11},\,\,C_{12}$ do not depend on
$l_1,\,l_2,\,n,\,k,\,(x,y),\,k',\,i$.

Let us now estimate the sum
$$
  \sum\limits_{k'=k'_m+1}^{k'_M}\left|\,\,
\int\limits_{l_{k,k'-1,i-1}}^{l_{k,k'-1,i}}
\widetilde{\sigma}(l)\;
\widetilde{P_x}(x,y;l)\;l^{k'-1}\,dl-
\int\limits_{l_{k,k'-1,-i}}^{l_{k,k'-1,1-i}}
\widetilde{\sigma}(l)\;
\widetilde{P_x}(x,y;l)\;l^{k'-1}\,dl\right|\leq
$$
$$
\leq  \sum\limits_{k'=k'_m+1}^{k'_M}\left|\,\,
\frac {C_{11} \,l_1^k}{nk'}
\left(
l_{k,k'-1,n}-l_{k,k'-1,-n}\right) \right|
 +   \sum\limits_{k'=k'_m+1}^{k'_M}\left|
\frac {C_{12} \,l_1^k}{2n(k')^2} \right|\leq
$$
\begin{equation}\label{22}
\leq\frac {C_{11} \,l_1^k}{nk(1-\gamma)}
  \sum\limits_{k'=k'_m+1}^{k'_M}
\left(
l_{k,k'-1,n}-l_{k,k'-1,-n}\right)
 +
\frac {C_{12} \,l_1^k}{2nk(1-\gamma)}\sum\limits_{k'=k'_m+1}^{k'_M}
\frac 1{k'}
\leq
\frac {C_{13} \,l_1^k}{nk} (l_2-l_1),
\end{equation}
where  $C_{13}$ does not depend on $l_1,\,l_2,\,(x,y),\,k,\,i$.
Let us clarify the bound for the second
summand. For sufficiently large $k$, we  have
$C_0+1\leq  k\left(1-\log_{l_2}l_1\right) $,
and thus
$$
\sum\limits_{k'=k'_m+1}^{k'_M}
\frac 1{k'}
\leq \frac{ N(x,y;k)}  {k'_m+1}  \leq
\frac {k\left(1-\log_{l_2}l_1\right)+C_0+1 } {k(1-\gamma)}
\leq
$$
$$
\leq\frac {2k\left(\log_{l_2}l_2-\log_{l_2}l_1\right)} {k(1-\gamma)}
\leq  C_{14} (l_2-l_1),
$$
where $C_{14}$ depends on $l_*,\,l_{**}$ only, which
proves the bound (\ref{22}).
Then we have
$$
\left|\sum\limits_{i=1}^{n}
c_i  \sum\limits_{k'=k'_m+1}^{k'_M}\left(\,\,
\int\limits_{l_{k,k'-1,i-1}}^{l_{k,k'-1,i}}
\widetilde{\sigma}(l)\;
\widetilde{P_x}(x,y;l)\;l^{k'-1}\,dl-
\int\limits_{l_{k,k'-1,-i}}^{l_{k,k'-1,1-i}}
\widetilde{\sigma}(l)\;
\widetilde{P_x}(x,y;l)\;l^{k'-1}\,dl\right)\right|^2\leq
$$
$$
\leq\left|\sum\limits_{i=1}^{n}
|c_i|
\frac {C_{13} \,l_1^k}{nk} (l_2-l_1)
\right|^2=
\frac {C_{13}^2 \,l_1^{2k}}{n^2k^2} (l_2-l_1)^2
          \left|\sum\limits_{i=1}^{n}
|c_i|
\right|^2\leq
\frac {C_{13}^2 \;l_1^{2k}}{n\,k^2} (l_2-l_1)^2
          \left(\sum\limits_{i=1}^{n}|c_i|^2
\right).
$$

It is clear that the first two summands in
(\ref{19,5}) can be estimated in a similar way, and therefore
 the following inequality holds  for any point $(x,y)\in R_k$:
$$
\left|\int\limits_{l_1}^{l_2}
{\cal P}\,dl\right|^2\leq
\frac {C_{15}^2 l_1^{2k}}{nk^2} (l_2-l_1)^2
          \left(\sum\limits_{i=1}^{n}|c_i|^2
\right),
$$
where $C_{15}$ depends  on  $l_*,\,l_{**}$ only. Then
$$
\int\!\!\!\int\limits_{{\cal D}_2}\left|\,
\int\limits_{\lambda_1}^{\lambda_2}\frac{\sigma(\mu)\sqrt{\mu}}
{2\sqrt{1-\mu}}
\frac {\d P}{\d x}(x,y;\mu)\;\f(P(x,y;\mu),\mu)
\,d\mu\right|^2dxdy
=\sum\limits_{k=0}^{+\infty}
\int\!\!\!\!\int\limits_{R_k}
\left|\int\limits_{l_1}^{l_2}
              {\cal P} \,dl
\right|^2dxdy\leq
  $$
$$
\leq\sum\limits_{k=k_0}^{+\infty}  S(R_k)
\frac {C_{15}^2 l_1^{2k}}{nk^2} (l_2-l_1)^2
          \left(\sum\limits_{i=1}^{n}|c_i|^2
\right) \leq
\sum\limits_{k=k_0}^{+\infty}
\frac {C_{16} }{nk^2} (l_2-l_1)^2
          \left(\sum\limits_{i=1}^{n}|c_i|^2
\right) \leq
$$
$$
\leq\frac {C_{16} }{n} (l_2-l_1)^2
          \left(\sum\limits_{i=1}^{n}|c_i|^2
\right) \sum\limits_{k=k_0}^{+\infty}
\frac {1}{k^2} \leq
\frac {C_{17} }{n} (l_2-l_1)^2
          \left(\sum\limits_{i=1}^{n}|c_i|^2
\right) =
C_{17} (l_2-l_1)^2
\|\te_{1,n}(x)\|^2_{L_2},
$$
where $S(R_k)$ stands for the area of the domain  $R_k$ and the constants
$C_{16},\,C_{17}$ depend  on $l_*,\,l_{**}$ and
$\sigma(\mu)$ only.
Since the quantity $C_{17} (l_2-l_1)^2$ does not depend on $n$, it follows that
the  inequality thus obtained  holds  for an arbitrary function
$\te_1(x)\in L_2(0,1)$,
$$
\int\!\!\!\int\limits_{{\cal D}_2}\left|\,
\int\limits_{\lambda_1}^{\lambda_2}\sigma(\mu)\;
 \frac {\d P}{\d x}(x,y;\mu)\;\f(P(x,y;\mu),\mu)\,
d\mu\right|^2\,dxdy\leq
C_{17} (l_2-l_1)^2
\|\te_1(x)\|^2_{L_2}.
$$
Obviously, the other integral in  (\ref{19}) can be  estimated  in a similar way,  and therefore
\begin{equation}\label{v'x1}
\left\| V'_x\right\|_{L_2
({\cal D}_2)}^2\leq
C_{18} (l_2-l_1)^2\|\te_1 (x)\|^2_{L_2}\leq
C_{19}(\lambda_2-\lambda_1)^2\|\te_1 (x)\|^2_{L_2},
\end{equation}
where $C_{19}$ does not depend on $\te_1,\,\lambda_1,\,\lambda_2.$
This proves the first assertion of the lemma.

To prove the other assertion of the lemma,  let
us estimate the quantity
$
 \left\| V'_x\right\|_{L_2
(D\setminus \overline{{\cal D}_2})}^2.
$
Let
$M(x,y)\in \left\{D\setminus \overline{{\cal D}_2}\right\}$,
i.e. let $M(x,y)\in \triangle AA_{\lambda_2}B$.
Let us draw the characteristic of the first family for equation (\ref{10})
through $M(x,y)$ until this characteristic meets the segment
$OA$
at some point $L(x-ay;0)$.
After this,  draw the characteristic of
 the other family for (\ref{10}) through $M(x,y)$ until it
meets the segment $OB$
at some point $K\left(\frac{x+ay}{1+a\alpha};
\alpha\frac{x+ay}{1+a\alpha}\right)$.
Note that,  by properties b), c), and d) of the function
$u(x,y;\te_1;\lambda)$,  the polygon
$T(x,y,\lambda)$
with the vertices at the points  $M,\,K,\,B,\,A,\,L$ satisfies the following
conditions:
$$
\int\limits_{\d T(x,y;\lambda)}
\frac{\d u}{\d x}\,dy+\frac 1{a^2}
\frac{\d u}{\d y}\,dx =
\int\!\!\!\!\!\!\int\limits_{T(x,y;\lambda)}
\left(\frac{\d^2 u}{\d x^2}-\frac 1{a^2}\frac{\d^2 u}{\d y^2}
\right)\,dx\,dy =0.
$$
This implies the relation
\begin{equation}
u|_M=u(x,y;\te_1;\lambda)=
-\frac a2 \int\limits_{x-ay}^{1/\alpha}
\f_0(t,\lambda)\,dt
+\frac a{2}\int\limits_{A}^{B}\te_1(t) \,dt+
\frac a{2}\int\limits_{1/\alpha}^{\frac{x+ay}{1+a\alpha}}
 \;\f(t,\lambda) \,dt,
\end{equation}
where
$\f_0(t,\lambda):=
\displaystyle\frac 1{a^2}\left.\frac{\d u}{\d y}\right|_{y=0}$.
Therefore, for almost all  $(x,y)\in \left\{D\setminus \overline{
{\cal D}_2}\right\}$, we have
$$
 V'_x=\int_{\lambda_1}^{\lambda_2}
u(x,y;\te_1;\lambda)'_{x}d\lambda =
$$
\begin{equation}
=
\int_{\lambda_1}^{\lambda_2}\frac a2\left(
 \f_0\left(x-ay,\lambda\right)
+
 \frac{1}{1+a\alpha} \;\f\left(
\frac{x+ay}{1+a\alpha},\lambda\right)\right)d\lambda.
\end{equation}
Then
$$
\left\| V'_x\right\|^2_{L_2(D\setminus
\overline{{\cal D}_2})}=
\int\!\!\!\!\!\int\limits_{D\setminus \overline{{\cal D}_2}}\left\{
\int_{\lambda_1}^{\lambda_2}\frac a2\left(
 \f_0\left(x-ay,\lambda\right)
+   \frac{1}{1+a\alpha}\; \f\left(
\frac{x+ay}{1+a\alpha},\lambda\right)\right)d\lambda
\right\}^2 \,dx\,dy\leq
$$
$$
\leq(\lambda_2-\lambda_1)
\int\!\!\!\!\!\int\limits_{D\setminus \overline{{\cal D}_2}}
\int_{\lambda_1}^{\lambda_2}\frac {a^2}4\left|
 \f_0\left(x-ay,\lambda\right)
+   \frac{1}{1+a\alpha} \;\f\left(
\frac{x+ay}{1+a\alpha},\lambda\right)\right|^2d\lambda
 \,dx\,dy\leq
$$
$$\leq M_0(\lambda_2-\lambda_1)
\int_{\lambda_1}^{\lambda_2}\left\{
\int\!\!\!\!\!\int\limits_{D\setminus \overline{{\cal D}_2}}
\left|
 \f_0\left(x-ay,\lambda\right)\right|^2
+   \left| \f\left(
\frac{x+ay}{1+a\alpha},\lambda\right)\right|^2
 \,dx\,dy\right\} \,d\lambda\leq
$$
\begin{equation}
\leq M_1(\lambda_2-\lambda_1)
\int_{\lambda_1}^{\lambda_2}\left\{
\left\|
 \f_0\right\|^2_{L_2(A^1(\lambda_2)A)}
+   \left\| \f\right\|^2_{L_2(B^2(\lambda_2)B
)} \,\right\} d\lambda\leq
 M_2(\lambda_2-\lambda_1)^2\|\te_1\|^2_{L_2},
 \end{equation}
where the positive constants  $M_i, i=0,1,2$ do not depend on
 $\lambda_1$, $\lambda_2$ and $\te_1$.
Taking (\ref{v'x1}) into account, we obtain
$$
 \left\| V'_x\right\|_{L_2(D)}^2\leq M_3
 (\lambda_2-\lambda_1)^2\|\te_1\|^2_{L_2},
$$
where the constant $M_3$ does not depend on
 $\lambda_1$, $\lambda_2$ and $\te_1$.
Obviously,  for
$
 \left\| V'_y\right\|_{L_2
(D\setminus \overline{{\cal D}_2})}^2,
$
we have a  similar bound.
This proves the lemma.
\end{proof}

\begin{proof}[Proof of Theorem \ref{teo:odinn}]
Lemma \ref{lem:odin} implies that the function
$\|U(x,y;\te_1;\sigma;\lambda)\|_{1}$ is  continuous
and even the absolute continuous with respect to  $\lambda$.

Let us now prove  that formula (\ref{11}) holds  for our function $U$
for any $\lambda_1,\,\lambda_2\in (0,(1+\alpha^2)^{-1})$.
Since
$$
 U(x,y;\te_1;\sigma;\lambda_2)-U(x,y;\te_1;\sigma;\lambda_1) =
\int_{\lambda_1}^{\lambda_2}
\sigma(\lambda) \;u(x,y;\te_1;\lambda) \,d\lambda
$$
and
$$
d U(x,y;\te_1;\sigma;\lambda) = \sigma(\lambda)
 \;u(x,y;\te_1;\lambda) \,d\lambda ,
$$
it follows that we are to prove the relation
$$
A\,\int_{\lambda_1}^{\lambda_2}
\sigma(\lambda) \;u(x,y;\te_1;\lambda) \,d\lambda =
\int_{\lambda_1}^{\lambda_2}\lambda\;
\sigma(\lambda) \;u(x,y;\te_1;\lambda) \,d\lambda ,
$$
or, equivalently, after applying  the
Laplace operator to this relation, to prove the
equality
$$
\frac{\d^2}{\d y^2} \,\int_{\lambda_1}^{\lambda_2}
\sigma(\lambda) \;u(x,y;\te_1;\lambda) \,d\lambda = \Delta
\int_{\lambda_1}^{\lambda_2}\lambda\,
\sigma(\lambda) \;u(x,y;\te_1;\lambda) \,d\lambda
$$
in the space  $W^{-1}_2(D)$.
This means that, for an arbitrary $g\in C_0^{\infty}(D)$, we must prove that
$$
\int\limits_{D}\!\!\!\int\left\{
\frac{\d^2}{\d y^2} \,\int_{\lambda_1}^{\lambda_2}
\sigma(\lambda) \;u(x,y;\te_1;\lambda) \,d\lambda \cdot g  - \Delta
\int_{\lambda_1}^{\lambda_2}\lambda\,
\sigma(\lambda) \;u(x,y;\te_1;\lambda) \,d\lambda \cdot g \right\}dxdy =0.
$$

Since  the function
$\sigma(\lambda) u(x,y;\te_1 ;\lambda)\in
W^1_2(D\cap \{x>\varepsilon\})$, $\varepsilon>0,$
is a generalized solution of  (\ref{10}) by construction, it follows that
$$
\int\limits_{D}\sigma u\,\frac {\d^2 g}{\d z^2}\;dxdy=
\lambda\,\int\limits_{D}\sigma\, u\;\Delta g \,dxdy
$$
for any  $g\in C_0^{\infty}(D)$.  Therefore
$$
\int\limits_{D}\!\!\!\int \left\{
\frac{\d^2}{\d y^2} \,\int_{\lambda_1}^{\lambda_2}
\sigma(\lambda) \;u(x,y;\te_1;\lambda) \,d\lambda \cdot g  - \Delta
\int_{\lambda_1}^{\lambda_2}\lambda\;
\sigma(\lambda) \;u(x,y;\te_1;\lambda) \,d\lambda \cdot g \right\}dxdy=
$$
$$
=\int\limits_{D}\!\!\!\int \left\{
 \int_{\lambda_1}^{\lambda_2}
\sigma(\lambda) \;u(x,y;\te_1;\lambda) \,d\lambda \cdot
\frac{\d^2 g}{\d y^2}  -
\int_{\lambda_1}^{\lambda_2}\lambda\;
\sigma(\lambda) \;u(x,y;\te_1;\lambda) \,d\lambda \cdot
\Delta g \right\}dxdy =
$$
$$
=\int\limits_{D}\!\!\!\int  \int_{\lambda_1}^{\lambda_2}\left\{
\sigma(\lambda) \;u(x,y;\te_1;\lambda)  \cdot
\frac{\d^2 g}{\d y^2}  -
\lambda\;
\sigma(\lambda) \;u(x,y;\te_1;\lambda) \cdot
\Delta g \right\} d\lambda \,dxdy =
$$
$$
= \int_{\lambda_1}^{\lambda_2} \int\limits_{D}\!\!\!\int \left\{
\sigma(\lambda) \;u(x,y;\te_1;\lambda)  \cdot
\frac{\d^2 g}{\d y^2}  -
\lambda \;
\sigma(\lambda) \;u(x,y;\te_1;\lambda) \cdot
\Delta g \right\}  dxdy\,d\lambda =0.
$$

 This completes the proof of  Theorem \ref{teo:odinn}.
\end{proof}

Let $\te_2(x)\in L_2(0,1/\alpha)$ be an arbitrary function.
In just the same way in which the function
$u(x,y;\te_1;\lambda)$ was constructed for $\lambda\in (0,\,(1+\alpha^2)^{-1})$,
we construct the function
$v(x,y;\te_2;\lambda)$ for $\lambda\in ((1+\alpha^2)^{-1},\,1)$, namely,
the function $v$ is the generalized solution of
(\ref{10}) belonging  to the space
$W^1_2(D\cap \{y<1-\varepsilon\})$ for any $0<\varepsilon <1$,
and satisfying  the relations
$v|_{\partial D}=0$ and $v_y|_{[OA]}=\theta_2$.
Consider the function
 \begin{equation} \label{24}
w(x,y;
\lambda):=\left\{
\begin{array}{ll}
u(x,y;\te_1;\lambda),
&\mbox{ for }  0<\lambda   <(1+\alpha^2)^{-1}, \\
v(x,y;\te_2;\lambda),
&\mbox{ for } (1+\alpha^2)^{-1}<\lambda <1.
\end{array}
\right.
 \end{equation}

\begin{cor}\label{cor:odin-tri}
If $\sigma (\lambda)$ is an arbitrary function in $C^1(0,1)$,
whose support $\supp \sigma $ belongs to
 the union of the intervals
$(0,(1+\alpha^2)^{-1} )$ and $((1+\alpha^2)^{-1},1)$, then
the function
 \begin{equation}\label{25}
U(x,y;\lambda):=
\int\limits_{0}^{\lambda}\sigma (\mu)w(x,y;\mu)d\mu,\quad
\lambda\in [0,\,1],
 \end{equation}
with values in the space $\Wp $
is a differential solution of $(\ref{9})$.
\end{cor}

\begin{cor}[ existence theorem for the absolutely continuous spectrum]
\label{cor:odin-chetyre}
Let $H_{0}$ be the closure of the linear span of
all differential solutions of the form $(\ref{25})$ in $\Wp(D)$.
 Theorem {\rm\ref{teo:odinn}}  proved above implies that
  the spectrum of the operator $A$ is absolutely continuous on $H_{0}$,
$\sigma (A|_{H_0})=\sigma_{ac}(A|_{H_0})$.
\end{cor}

\begin{ex}\rm
In some works, the  class of  bounded convex domains
 $\Omega$ under consideration is defined by the condition that
$\G=\partial\Omega=\cup_{j=1}^n \G_j$, $\G_j\in C^4$, and
$\G_j$ is either a line segment  or  has
 positive curvature at every point, including the endpoints
(the value of curvature at the endpoints is understood as the limit of the curvature
at the interior points).
Let a set $\{\a_1,\dots,\a_m\}\ss \left[0,\frac \pi 2\right]$
consist of angles between the axis   $Ox$
and all one-side tangents  at the endpoints
$\G_j$, $j\in \{1,\dots,n\}$.
Let $0\le \a_1< \dots < \a_m\le \frac \pi 2$. Write
$$
\a=\a(\la)=\arccos (\sqrt{\la})\in \left[0,\frac \pi 2\right],
\qquad \la=\cos^2 \a,
$$
$$
\xi=\xi(X,\a)=x \sin \a + y \cos \a,\qquad
\eta=\eta(X,\a)=x \sin \a - y \cos \a,\qquad X=(x,y).
$$
Consider  domains of this kind for which
there are  $\a^*\in \{\a_1,\dots,\a_m\}$ and $X^1$, $X^2 \in \G$
such that
$$
\left\{%
\begin{array}{l}
  \xi(X^1,\a^*)= \xi(X^2,\a^*)   \mbox{ or } \eta(X^1,\a^*)= \eta(X^2,\a^*),\\
    \{X\in \R^2 \,|\, \xi(X,\a^*)=\xi(X^j,\a^*)\}\cap \Omega=
    \emptyset,\quad j=1,2, \\
     \{X\in \R^2 \,|\, \eta(X,\a^*)=\eta(X^j,\a^*)\}\cap \Omega=
    \emptyset,\quad j=1,2. \\
\end{array}%
\right.
$$
As it was noted above, the number  $\la\in (0,1)$
is an eigenvalue of the operator  $A$ corresponding to
the domain $\Omega$ if and only if the hyperbolic equation
(\ref{10}) has a nontrivial generalized solution in
$\Wp(\Omega)$.

In particular, let  $\Omega$ be the quadrangle with
vertices at the points $O(0,0)$, $A\left(\frac 13,\frac 13\right)$,
 $B\left(\frac 12,1\right)$ è $C(0,1)$. Then $n=m=4$ and
the above angles are: $0$, $\pi/4$, $\arctan 4$, $\pi/2$.
Let
 $\a^*=0$, $X^1=C$, $X^2=B$. Then $\xi(X^1,\a^*)=
 \xi(X^2,\a^*)$ and, as can readily be seen, the four
 lines
$\xi(X,\a^*)=\xi(X^j,\a^*)$, $\eta(X,\a^*)=\eta(X^j,\a^*)$,
$j=1,2,$ do not intersect the domain $\Omega$.
Thus, the domain in question belongs to the class of domains
described above.

It can immediately be proved  that the function
$$
u(x,y):=C^*\cdot
\left\{%
\begin{array}{ll}
    x, & \hbox{for $(x,y)\in \triangle COM$;} \\
    1-y, & \hbox{for $(x,y)\in \triangle CMB$;} \\
    y-x, & \hbox{for $(x,y)\in \triangle OMA$;} \\
    \frac 12 (y+1)-2x, & \hbox{for $(x,y)\in \triangle MAB$;} \\
\end{array}%
\right.
$$
where $M$ has coordinates $\left(\frac 14,\frac 12\right)$ and
$C^*\ne 0$ is an arbitrary constant, belongs to the space
$\Wp(\Omega)$ and is the generalized solution of the
Dirichlet problem for equation  (\ref{10}) for
 $\la=\cos^2(\arctan 2)=\frac 15$.
Thus, $\la=\frac 15$ is an eigenvalue of the operator $A$ corresponding
to the domain  $\Omega$.
\end{ex}

\section{Solutions of the nonstationary Sobolev equation}\label{sec:vtoraja}

Using the  differential solutions
$U(x,y;\lambda)$ of the form (\ref{25}), which are constructed
above, we can write out some exact solutions
of the nonstationary Sobolev equation. Namely,
the following assertion is an immediate corollary to
Theorem \ref{teo:odinn}.
\begin{teo}\label{teo:dvva}
Let
 \begin{equation}\label{26}
p_0(x,y):=
\int\limits_{0}^{1}
\sigma_0(\lambda)\,w_0(x,y;\lambda)\,d\lambda,\quad
p_1(x,y):=
\int\limits_{0}^{1}
\sigma_1(\lambda)\,w_1(x,y;\lambda)\,d\lambda,
 \end{equation}
where $w_0,\,w_1$ are constructed in the  way described above
 from functions
$\te^{(0)}_1,$ $\te^{(0)}_2$, $\te^{(1)}_1$, $\te^{(1)}_2$,
respectively, and the supports $\supp \sigma_i $
are contained in the union of the intervals
$(0,(1+\alpha^2)^{-1} )$ and $((1+\alpha^2)^{-1},1)$.
Then the function
 \begin{equation}\label{27}
p(x,y;t):=
\int\limits_{0}^{1}   \cos(\sqrt{\lambda}\; t)
\,\sigma_0(\lambda)\,w_0(x,y;\lambda)\,d\lambda
+
\int\limits_{0}^{1}   \frac{\sin(\sqrt{\lambda}\;t)}{\sqrt{\lambda}}\,
\sigma_1(\lambda)\,w_1(x,y;\lambda)\,d\lambda,
 \end{equation}
is a solution to  problem $(\ref{1})-(\ref{3})$
and  $\|p(x,y;t)\|_{L_2(D)}\to 0$ as $t\to \infty$.
\end{teo}

Note that here we have $p_0(x,y),\,\, p_1(x,y)\in H_0$.

It follows from the results of the papers \cite{Tro94},
\cite{Tro98}, and \cite{Tro99} that, for any point
$(x,y)\in D$, the function $w(x,y;\lambda)$ is an absolutely
continuous function of the variable
$\lambda$ on the closed interval  $[0,\,1]$ (see the property c)
of the function $u(x,y;\te_1;\lambda)$), and therefore
$$
p(x,y;t)\to 0\quad (t\to \infty )
$$
for any point $(x,y)\in D$.

\begin{teo}\label{teo:trrri}
Let
$\te^{(0)}_1,\,\te^{(1)}_1\in C_0^{\infty}[0,1]$,
$\te^{(0)}_2,\,\te^{(1)}_2\in C_0^{\infty}[0,1/{\alpha}]$,
 $\sigma_i(\lambda)\in C_0^{\infty}[0,1]$,
let the supports $\supp \sigma_i $ are contained in the
union of the intervals
$(0,(1+\alpha^2)^{-1} )$ and $((1+\alpha^2)^{-1},1)$, $i=0,1$.
In this case,  for any $n=1,2,...$, there exist a constant ${\cal K}_n$
independent of $t$ and such that the bound
\begin{equation}\label{oc}
\|p(x,y;t)\|_{L_2(D)}\leq \frac {{\cal K}_n}{t^n}
\end{equation}
holds for any $t \in (0,\,\infty)$,
where $p(x,y;t)$ stands for  the solution of  problem
$(\ref{1})-(\ref{3})$ of the form $(\ref{27})$.
\end{teo}

 \begin{proof}
Obviously, it  suffices to consider the case in which
$\te^{(1)}_1\equiv 0$,
$\te^{(0)}_2\equiv 0,\,\te^{(1)}_2\equiv 0$,
 $\sigma_1(\lambda)\equiv 0$,
 $\supp \sigma_0 \subset [\lambda_*,\,\lambda_{**}]
\subset (0,(1+\alpha^2)^{-1} )$, i.e.
\begin{equation}\label{28}
p(x,y;t):=
\int\limits_{0}^{1}   \cos(\sqrt{\lambda}\; t)\,
\sigma_0(\lambda)\;u(x,y;\lambda)\,d\lambda =
\int\limits_{0}^{(1+\alpha^2)^{-1}}   \cos(\sqrt{\lambda}\; t)\,
\sigma_0(\lambda)\;u(x,y;\lambda)\,d\lambda,
 \end{equation}
where $u(x,y;\lambda)=u(x,y;\te^{(0)}_1;\lambda)$.
Under the assumptions formulated in the theorem,
$u(x,y;\lambda)$ is infinitely  differentiable with respect to
 $x,y,\lambda$ everywhere inside the prisme
$D\times (0,(1+\alpha^2)^{-1})$
and can be continuously extended with all its derivatives,
to the entire surface of the prisme, except for
the line $x=0$, $y=0$, and the corresponding function
$\widetilde{\f}(x,l)$ (see (\ref{phi}))
is infinitely  differentiable with respect to
$x,l$ everywhere inside the half-strip
 $(0,1/{\alpha})\times (1,\infty)$.
Let  $l_1:=l(\lambda_*)$, $l_{2}:=l(\lambda_{**})$. Then we obviously
have
$1<l_1 \leq l_{2}<+ \infty$.

 \begin{lem}\label{lem:dva-tri}
For any  $n=0,1,2,...,$
the following bounds hold:
\begin{equation}\label{ocphi}
\left|\frac {\d^n\widetilde{\f}}{\d l^n}\right|\leq
C(n)k^n\,l^k, \quad       \frac 1{\alpha l^{k+1}}<x\leq \frac 1{
\alpha l^{k}},
\quad l_1\leq l\leq l_{2},
\end{equation}
where $k=1,2,...,$ and the constant $C(n)$ does not depend on $k$
and $(x,l)$.
 \end{lem}

\begin{proof}
Since   $1<l_1 \leq l\leq l_{2}<+ \infty$, it follows that,  obviously, it suffices  to prove the
above bounds for the function
$\Phi(x,l):=\widetilde{\f}(x,l)\frac {l-1}{2\alpha l}$.

Let $-\frac 12\leq\tau < \frac 12$
and let
$$
C_{\tau}:= \Phi\left(
\frac 1{\alpha }
\left(\frac 12-\tau\right)+\frac 1{\alpha l_{*}}
\left(\frac 12+\tau\right),l_{*}\right).
$$
Introduce the functions
$$
\widehat{x}_{k,\tau}(l)=
\frac 1{\alpha l^k}\left(\frac 12-\tau\right)
+\frac 1{\alpha
l^{k+1}}\left(\frac 12+\tau\right),\qquad k=0,1,2,....
$$

Using (\ref{x(k,j)}) and (\ref{Phi}) it can readily be proved  that
\begin{eqnarray*}
1)\quad && \Phi\left(\frac {x}l,l\right)=l\Phi(x,l),\qquad (x,l)\in
(0,1/{\alpha})\times (1,\infty),\\
2)\quad && \left.\Phi(x,l)\right|_{x=
\widehat{x}_{k,\tau}(l)}=C_{\tau}l^k,
\qquad
k=0,1,2,...,
\end{eqnarray*}
which immediately implies the assertion  of the lemma for $n=0$.

Further, it follows from  property 1)  that
$$
\frac{\d\Phi}{\d x}
\left(\frac {x}l,l\right)=l^2\,\frac{\d\Phi}{\d x}
(x,l).
$$
Let
$$
M_1=\max\limits_{(x,l)\in \overline{I_1}}
\left|\frac{\d\Phi}{\d l}\right|,\quad\mbox{ where }
I_1:=\left\{(x,l)\,\left|\, \frac 1{\alpha l}< x\leq
\frac 1{\alpha },\,\,l_1\leq l\leq l_{2}\right.\right\}.
$$
Then the following bound holds in the
 $k$-th strip
$\displaystyle
I_k:=\left\{(x,y)\,\left|\, \frac 1{\alpha l^{k+1}}< x\leq\frac 1{
\alpha l^{k}},\quad l_1\leq l\leq l_{2}\right.\right\}$:
$$
\left|\frac{\d\Phi}{\d x}(x,l)\right|\leq M_1\,l^{2k}.
$$
Let $(x_0,l_0)$ be an arbitrary point of the
$k$-th strip, and let $\tau_0$ be such that
$$
x_0=\widehat{x}_{k,\tau_0}(l_0).
$$
Then, on  one hand,
$$
\left.\frac{d\Phi(\widehat{x}_{k,\tau_0}(l),\,l)}{dl}
\right|_{l_0}=C_{\tau_0}\,k\,l_0^{k-1},
$$
and, on the other hand,
$$
\frac{d\Phi(\widehat{x}_{k,\tau_0}(l),\,l)}{dl}=
\left.\left[\frac{\d\Phi}{\d x}\cdot\frac{d\widehat{x}_{k,\tau_0} }{d l}+
\frac{\d\Phi}{\d l}\right]\right|_{
x=\widehat{x}_{k,\tau_0}(l)},
$$
and therefore
$$
\left.\frac{d\Phi(\widehat{x}_{k,\tau_0}(l),\,l)}{dl}\right|_{l_0}=
\left.\frac{\d\Phi}{\d x}
\right|_{(x_0,l_0)}\cdot
\left(
\frac {-k}{\alpha l_0^{k+1}}\left(\frac 12-\tau_0\right)
+\frac {-k-1}{\alpha
l_0^{k+2}}\left(\frac 12+\tau_0\right)\right)
+
\left.\frac{\d\Phi}{\d l}\right|_{(x_0,l_0)}.
$$
Hence, for any point  $(x_0,l_0)\in I_k$,
$$
\left|\left.\frac{\d\Phi}{\d l}\right|_{(x_0,l_0)}\right|=
\left|C_{\tau_0}\,k\,l_0^{k-1}+
\left.\frac{\d\Phi}{\d x}
\right|_{(x_0,l_0)}\cdot
\left(
\frac {k}{\alpha l_0^{k+1}}\left(\frac 12-\tau_0\right)
+\frac {k+1}{\alpha
l_0^{k+2}}\left(\frac 12+\tau_0\right)\right)\right|\leq
$$
$$
\leq\left|C_{\tau_0}\,k\,l_0^{k-1}\right|+
M_1\,l_0^{2k}
\cdot \left|
\frac {k}{\alpha l_0^{k+1}}
+\frac {k+1}{\alpha
l_0^{k+2}}\right|\leq  C(1)\,k\,l_0^k.
$$
which implies the assertion of the lemma for $n=1$.
To obtain the desired bound for $n=2$,
it  suffices to note that
$$
\frac{\d^2\Phi}{\d x^2}
\left(\frac {x}l,l\right)=l^3\,\frac{\d^2\Phi}{\d x^2}(x,l)
$$
and to consider  the second derivative of
the function
$\Phi(\widehat{x}_{k,\tau_0}(l),\,l)$, respectively,  and so on.
\end{proof}


Let us prove now the estimation (\ref{oc}) for $n=1$. We have
$$
p(x,y;t)=
\int\limits_{0}^{1}   \cos(\sqrt{\lambda}\; t)\,
\sigma_0(\lambda)\;u(x,y;\lambda)\,d\lambda
=\int\limits_{0}^{1}  \cos (\nu t)\, 2\nu\;
\sigma_0(\nu^2)\;u(x,y;\nu^2)\,d\nu =
$$
\begin{equation}\label{30}
=-\frac 1t \int\limits_{0}^{1}\sin(\nu t) \;\left(2\nu
\sigma_0(\nu^2)\right)'_{\nu}u(x,y;\nu^2)\,d\nu
-\frac 1t \int\limits_{0}^{1}\sin (\nu t)\, 2\nu\;
\sigma_0(\nu^2)\;u'_{\nu}(x,y;\nu^2) \,d\nu.
\end{equation}
We claim that  the $L_2$-norm of the first integral
is bounded,
$$
\int\limits_{D}\!\!\!\int\left\{
\int\limits_{0}^{1}\sin(\nu t)\; \left(2\nu
\sigma_0(\nu^2)\right)'_{\nu}u(x,y;\nu^2)d\nu\right\}^2dx\,dy\leq
$$
$$
\leq\int\limits_{D}\!\!\!\int
\int\limits_{0}^{1}\left| \left(2\nu\;
\sigma_0(\nu^2)\right)'_{\nu}
u(x,y;\nu^2)\right|^2 d\nu\,dx\,dy
=\int\limits_{0}^{1}\int\limits_{D}\!\!\!\int
\left| \left(2\nu\;
\sigma_0(\nu^2)\right)'_{\nu}
u(x,y;\nu^2)\right|^2 dx dy\,d\nu \leq K_1,
$$
where $K_1$ does not depend on $t$.
Further, for the second integral in (\ref{30}), we obtain
$$
\int\limits_{D}\!\!\!\int
\left| \int\limits_{0}^{1}\sin(\nu t)\, 2\nu\;
\sigma_0(\nu^2)\;u'_{\nu}(x,y;\nu^2) d\nu\right|^2 dx \,dy=
$$
$$
=
\int\limits_{D}\!\!\!\int
\left| \int\limits_{{\lambda_*}
}^{{\lambda_{**}}}\sin (\sqrt{\mu}\; t)\,
\sigma_0(\mu)\;u'_{\mu}(x,y;\mu)\;
2\sqrt{\mu}\;
d\mu\right|^2 dx\, dy\leq
$$
$$
\leq ({\lambda_{**}}-{\lambda_{*}})
\int\limits_{D}\!\!\!\int
 \int\limits_{{\lambda_*}
}^{{\lambda_{**}}}\left|
\sigma_0(\mu)\;u'_{\mu}(x,y;\mu)\;
2\sqrt{\mu}\;
\right|^2 \,d\mu \,dx\, dy=
$$
$$
\leq ({\lambda_{**}}-{\lambda_{*}})
\int\limits_{{\cal D}_{**}}\!\!\!\int
 \int\limits_{{\lambda_*}
}^{{\lambda_{**}}}\left|
\sigma_0(\mu)\;u'_{\mu}(x,y;\mu)\;
2\sqrt{\mu}
\right|^2 \,d\mu \,dx \,dy+
$$
$$
+ ({\lambda_{**}}-{\lambda_{*}})
\int\limits_{D\setminus\overline{{\cal D}_{**}}}\!\!\!\int
 \int\limits_{{\lambda_*}
}^{{\lambda_{**}}}\left|
\sigma_0(\mu)\;u'_{\mu}(x,y;\mu)\;
2\sqrt{\mu}
\right|^2 \,d\mu\, dx \,dy,
$$
where
${\cal D}_{**}:=D\cap \{a_{**}y>x+a_{**}-\frac 1{\alpha}\}$ for
 $a_{**}=\sqrt{\frac{\lambda_{**}}{1-\lambda_{**}}}$.
The existence of the last summand  is obvious,
and therefore we study the first summand by
 using the representation
 (\ref{13,5})  of the function $u(x,y;\mu)$ in the domain
${\cal D}_{**}$:
$$
\int\limits_{{\cal D}_{**}}\!\!\!\int
\int\limits_{{\lambda_*} }^{{\lambda_{**}}}\left|
\sigma_0(\mu)\,u'_{\mu}(x,y;\mu)\;
2\sqrt{\mu}
\right|^2 \,d\mu\, dx \,dy\leq
$$
$$
\leq 4\int\limits_{{\cal D}_{**}}\!\!\!\int
 \int\limits_{{\lambda_*}
}^{{\lambda_{**}}}\left|
\sigma_0(\mu)\;
2\sqrt{\mu}
\left(\frac {\sqrt{\mu}}{2\sqrt{1-\mu}}\right)'_{\!\!\mu}
\;
\int\limits_{P(x,y;\mu)}^{Q(x,y;\mu)}
\f(x',\mu)\,dx'
\right|^2 \,d\mu\, dx\, dy+
$$
$$
+4\int\limits_{{\cal D}_{**}}\!\!\!\int
 \int\limits_{{\lambda_*}
}^{{\lambda_{**}}}\left|
\sigma_0(\mu)\;
\frac {{\mu}}{\sqrt{1-\mu}}
\;Q'_{\mu}(x,y;\mu)\,
\f(Q(x,y;\mu),\mu)
\right|^2 \,d\mu\, dx\, dy+
$$
$$
+4\int\limits_{{\cal D}_{**}}\!\!\!\int
 \int\limits_{{\lambda_*}
}^{{\lambda_{**}}}\left|
\sigma_0(\mu)\;
\frac {{\mu}}{\sqrt{1-\mu}}
\,P'_{\mu}(x,y;\mu)\,
\f(P(x,y;\mu),\mu)
\right|^2 \,d\mu\, dx \,dy+
$$
\begin{equation}\label{me}
+4\int\limits_{{\cal D}_{**}}\!\!\!\int
 \int\limits_{{\lambda_*}
}^{{\lambda_{**}}}\left|
\sigma_0(\mu)\;
\frac {{\mu}}{\sqrt{1-\mu}}\;
\int\limits_{P(x,y;\mu)}^{Q(x,y;\mu)}\;
\f'_{\mu}(\widetilde x,\mu)\,d\widetilde x
\right|^2 \,d\mu\, dx\, dy.
\end{equation}

Let us prove the existence of the last integral.
To this end, it  suffices to prove
the convergence of the series
\begin{equation} \label{sum}
\sum\limits_{k=k_0}^{+\infty}
\int\!\!\!\!\int\limits_{R_k}
 \int\limits_{\lambda_{*}}^{\lambda_{**}}\left|\;
\int\limits_{P(x,y;\mu)}^{Q(x,y;\mu)} \f'_{\mu}(\widetilde x,\mu)
\,d\widetilde x\right|^2 d\mu \,dx\,dy,
\end{equation}
where the integrals are taken over trapezium
$$
R_k:=\left\{(x,y)\,\left|\,\frac 1{l_1^{k+1}}<x<\frac 1{l_1^{k}},\,\,0<y<
\alpha x \right.\right\}
$$
and $k_0$ is sufficiently large.
Passing to the new  variable of integration $l=l(\mu)$,
we obtain
$$
\int\limits_{\lambda_{*}}^{\lambda_{**}}\left|\;
\int\limits_{P(x,y;\mu)}^{Q(x,y;\mu)} \f'_{\mu}(\widetilde x,\mu)\,
d\widetilde x\right|^2 \,d\mu \leq
K_2 \int\limits_{l_{1}}^{l_{2}}\left|\;
\int\limits_{\widetilde{P}(x,y;l)}^{\widetilde{Q}(x,y;l)}
\widetilde {\f}'_{l}(\widetilde x,l)\,d\widetilde x\right|^2 \,dl\leq
$$
$$
\leq K_{2} \int\limits_{l_{1}}^{l_{2}}
\left|\widetilde{P}(x,y;l)-\widetilde{Q}(x,y;l)\right|
\int\limits_{\widetilde{P}(x,y;l)}^{\widetilde{Q}(x,y;l)}
|\widetilde{\f}'_{l}(\widetilde x,l)|^2\,d\widetilde x\,
dl,
$$
where $K_2$  stands for a positive constant depending  on
$\lambda_{*}$ and $\lambda_{**}$ only.
Let
$$
P_k:=\max\limits_{(x,y)\in R_k,\,l\in [l_1,l_2]} {\widetilde{P}(x,y;l)}
,\quad
Q_k:=\min\limits_{(x,y)\in R_k,\,l\in [l_1,l_2]} {\widetilde{Q}(x,y;l)}.
$$
It can readily be seen that
\begin{equation}\label{P_k,Q_k}
P_k=\frac{l_2+1}{2l_1^k},\quad
Q_k=\frac{l_2+1}{2l_1^{k+1}l_2},
\end{equation}
and therefore, for $(x,y)\in R_k$,
\begin{equation}\label{F}
\int\limits_{l_{1}}^{l_{2}}
\left|\widetilde{P}(x,y;l)-\widetilde{Q}(x,y;l)\right|
\int\limits_{\widetilde{P}(x,y;l)}^{\widetilde{Q}(x,y;l)}
|\widetilde{\f}'_{l}(\widetilde x,l)|^2\,d\widetilde x\,dl
\leq \frac {K_{3}}{l_1^k}
\int\limits_{Q_k}^{P_k}
\left\{\int\limits_{l_{1}}^{l_{2}}
|\f'_{l}(\widetilde x,l)|^2\,dl\right\}d\widetilde x.
\end{equation}
where $K_3$  stands for some positive constant depending  on
$\lambda_{*}$ and $\lambda_{**}$ only. Moreover, it also follows
from (\ref{P_k,Q_k}) that there are
positive integers $m_0$ and $r_0$ independent of $k$ and
such that
$$
l_1^{-(k+m_0)}\leq Q_k<P_k\leq l_1^{-(k-r_0)}
$$
for any $k$. Therefore, denoting the integrand
on the right-hand side of (\ref{F}) by ${\cal F}$, we see that
${\cal F}$ is nonnegative and
$$
\int\limits_{Q_k}^{P_k}{\cal F} d\widetilde x
\leq
\int\limits_{l_1^{-(k+m_0)}}^{l_1^{-(k+m_0-1)}}{\cal F}d\widetilde x+
\int\limits_{l_1^{-(k+m_0-1)}}^{l_1^{-(k+m_0-2)}}{\cal F}d\widetilde x+ ...+
\int\limits_{l_1^{-(k-r_0+1)}}^{l_1^{-(k-r_0)}}{\cal F}d\widetilde x.
$$
Then
$$
\sum\limits_{k=k_0}^{+\infty}
\int\!\!\!\!\int\limits_{R_k}
 \int\limits_{\lambda_{*}}^{\lambda_{**}}\left|
\int\limits_{P(x,y;\mu)}^{Q(x,y;\mu)} \f'_{\mu}(\widetilde x,\mu)
d\widetilde x\right|^2
\,d\mu \,dx\,dy\leq
$$
\begin{equation}\label{sum1}
\leq
(m_0+r_0+1)K_{3}\sum\limits_{k=k_0}^{+\infty}
 \frac {S_k}{l_1^k}
\int\limits_{l_1^{-(k+1)} }^{l_1^{-k}}\left\{
\int\limits_{l_{1}}^{l_{2}}
|\widetilde{\f}'_{l}(\widetilde x,l)|^2\,dl\right\} d\widetilde x,
\end{equation}
where $S_k$ stands for the area of $R_k$.

Let  $\widetilde x\in [l_1^{-(k+1)}, l_1^{-k}]$.
Denote by
$ l_{\widetilde x,i}$ the values of $l$
corresponding to the  points of intersection of the line
$x=\widetilde x$ with the curves
$x=\frac 1{l^{k-i+1}}$, $i=0,1,2,...$,
$$
\widetilde x=\frac 1{( l_{\widetilde x,0})^{k+1}}=\frac 1
{( l_{\widetilde x,1})^{k}}=
\frac 1{( l_{\widetilde x,2})^{k-1}}=\dots =
\frac 1{( l_{\widetilde x,i+1})^{k-i}}=\dots ,
$$
and by $N(\widetilde x,k)$ the number of intersections which
correspond to $ l_{\widetilde x,i}$ belonging to the  interval
$[l_1,l_2]$.
It can readily be seen  that
$N(\widetilde x,k)<\widehat \gamma k$ for sufficiently
large $k$, where
$0<\widehat \gamma<1$ does not depend on $k$ and $\widetilde x$.
Here we obviously have
$ l_{\widetilde x,0}\leq l_1\leq l_{\widetilde x,1}$ and
$$
\int\limits_{l_1^{-(k+1)} }^{l_1^{-k}}d\widetilde x
\int\limits_{l_{1}}^{l_{2}}
|\widetilde{\f}'_l(\widetilde x,l)|^2\,dl\leq
\int\limits_{l_1^{-(k+1)} }^{l_1^{-k}}\left\{
\sum\limits_{i=0}^{
N(\widetilde x,k)+1 }  \int\limits_{l_{\widetilde x,i} }
^{l_{\widetilde x,i+1}}
|\widetilde{\f}'_l(\widetilde x,l)|^2\,dl\right\}d\widetilde x.
$$
Let us now use the bound (\ref{ocphi}) for $(n=1)$. In this case,
$$
\int\limits_{l_1^{-(k+1)} }^{l_1^{-k}}\left\{
\sum\limits_{i=0}^{
N(\widetilde x,k)+1 }  \int\limits_{l_{\widetilde x,i} }
^{l_{\widetilde x,i+1}}
|\widetilde{\f}'_l(\widetilde x,l)|^2\,dl\right\}d\widetilde x
$$
$$
\leq
\int\limits_{l_1^{-(k+1)} }^{l_1^{-k}}\left\{
\sum\limits_{i=0}^{
N(\widetilde x,k)+1 }  \int\limits_{l_{\widetilde x,i} }
^{l_{\widetilde x,i+1}}
C^2(1)\,(k-i)^2\,l^{2(k-i)}\,dl\right\}d\widetilde x
$$
$$
\leq k^2\, C^2(1)
\int\limits_{l_1^{-(k+1)} }^{l_1^{-k}}\left\{
\sum\limits_{i=0}^{ N(\widetilde x,k)+1 }  (l_{\widetilde x,i+1})^{2(k-i)}
({l_{\widetilde x,i+1}}-{l_{\widetilde x,i} })\right\}d\widetilde x
$$
$$
 \leq k^2\, C^2(1)
\int\limits_{l_1^{-(k+1)} }^{l_1^{-k}}\left\{
( N(\widetilde x,k)+2)\,(l_{\widetilde x,0})^{2(k+1)}\right\}d\widetilde x
 \leq   k^2\, C^2(1)
\int\limits_{l_1^{-(k+1)} }^{l_1^{-k}}\left\{
(\widehat \gamma k+2)l_1^{2(k+1)}\right\}\,d\widetilde x
$$
$$
 \leq K_4\,k^3 l_1^k,
$$
where $K_4$ stands for  a positive constant
independent of $k$. This obviously implies that
 the series (\ref{sum1}) converges.

It is clear that the same arguments prove the convergence of
the first integral on the right-hand side of
inequality (\ref{me}), because the function
$\f(x,\mu)$ can not increase more rapidly
 than $\f'_{\mu}(x,\mu)$.

Consider now the second integral on the right-hand
side of (\ref{me}). Since
$$
{P'_{\mu}(x,y;\mu)}=\frac {\alpha x-y}{2\alpha}\,l'_{\mu},
$$
 we obviously have
$$
\int\limits_{{\cal D}_{**}}\!\!\!\int
 \int\limits_{{\lambda_*}
}^{{\lambda_{**}}}\left|
\sigma_0(\mu)\;
\frac {{\mu}}{\sqrt{1-\mu}}
{P'_{\mu}(x,y;\mu)}\;
\f(P(x,y;\mu),\mu)
\right|^2 \,d\mu \,dx \,dy
$$
$$
  \leq  K_{5}\int\limits_{{\cal D}_{**}}\!\!\!\int |x|
 \int\limits_{{\lambda_*}
}^{{\lambda_{**}}}   \left|
\f(P(x,y;\mu),\mu)
\right|^2 \,d\mu \,dx \,dy
$$
$$
\leq
K_{6}\int\limits_{{\cal D}_{**}}\!\!\!\int |x|
 \int\limits_{l_1}^{l_2}
\left|
\widetilde{\f}(\widetilde{P}(x,y;l),l)
\right|^2 \,dl \,dx \,dy,
$$
where $K_5$ and $K_{6}$ are some positive constants depending
on $\sigma_0$,  $\lambda_{*}$, and $\lambda_{**}$ only.
Using  the notation of  Section~\ref{sec:pervaja},
we shall now prove that the series
\begin{equation}\label{ryad}
\sum\limits_{k=k_0}^{+\infty}
\int\!\!\!\!\int\limits_{R_k} |x|
\int\limits_{l_1
}^{l_2}   \left|
\widetilde{\f}(\widetilde{P}(x,y;l),l)
\right|^2 \,dl \,dx \,dy
\end{equation}
converges, where $k_0$ is sufficiently large.
For any point $(x,y)\in R_k$, we  have
$$
\int\limits_{l_1}^{l_2}
\left|\widetilde{\f}(\widetilde{P}(x,y;l),l)\right|^2 \,dl =
\int\limits_{l_1}^{l_{k,k'_M,n}}
\left|\widetilde{\f}(\widetilde{P}(x,y;l),l)\right|^2 \,dl
+
\int\limits_{l_{k,k'_M,n}}^{l_{k,k'_M-1,n}}
\left|\widetilde{\f}(\widetilde{P}(x,y;l),l)\right|^2 \,dl
+...+
$$
$$
+
\int\limits_{l_{k,k'_m,n}}^{l_2}
\left|\widetilde{\f}(\widetilde{P}(x,y;l),l)\right|^2 \,dl
\leq
$$
$$
\leq C(0)\left\{
\int\limits_{l_{k,k'_M+1,n}}^{l_{k,k'_M,n}} l^{2k'_M}dl +
\int\limits_{l_{k,k'_M,n}}^{l_{k,k'_M-1,n}} l^{2k'_M-2}dl +...+
\int\limits_{l_{k,k'_m,n}}^{l_{k,k'_m-1,n}}l^{2k'_m-2} dl
\right\}\leq
$$
$$
\leq C(0)(l_2-l_1)
\left((l_{k,k'_M,n})^{2k'_M} +
(l_{k,k'_M-1,n})^{2k'_M-2} +...+
(l_{k,k'_m-1,n})^{2k'_m-2} \right)=
$$
$$=
 C(0)(l_2-l_1)
\left( P(x,y;{l_{k,k'_M,n}})^{-2} +
P(x,y;{l_{k,k'_M-1,n}})^{-2} +...+
P(x,y; {l_{k,k'_m-1,n}})^{-2} \right)\leq
$$
$$
\leq C(0)(l_2-l_1)(N(x,y;k)+1)\frac {4\alpha^2}{(\alpha x+y)^2}
\leq  \frac {K_{7}k}{x^2},
$$
where  $K_{7}$ does not depend on $k$, and therefore the series (\ref{ryad})
converges indeed.

Thus, Theorem \ref{teo:trrri} is proved for $n=1$.
To prove the validity of the bound  (\ref{oc}) for $n=2$,
note that
$$
U(x,y;t)=
$$
$$
=-\frac 1t \int\limits_{0}^{1}\sin(\nu t)\, \left(2\nu\;
\sigma_0(\nu^2)\right)'_{\nu}\;u(x,y;\nu^2)\,d\nu
-\frac 1t \int\limits_{0}^{1}\sin(\nu t)\, 2\nu\;
\sigma_0(\nu^2)\;u'_{\nu}(x,y;\nu^2) \,d\nu=
$$
$$
=-\frac 1{t^2}
\left(
\int\limits_{0}^{1}   \cos({\nu} t)\,
\left(2\nu\;\sigma_0(\nu^2)\right)''_{\nu\nu}\,u(x,y;\nu^2)\,d\nu+
2 \int\limits_{0}^{1}  \cos(\nu t)\,\left( 2\nu\;
\sigma_0(\nu^2)\right)'_{\nu}\,u'_{\nu}(x,y;\nu^2)\,d\nu\right. +
$$
$$
+\left. \int\limits_{0}^{1}\cos(\nu t)\, 2\nu\;
\sigma_0(\nu^2)\;u''_{\nu\nu}(x,y;\nu^2)\,d\nu\right).
$$
The existence of the first two integrals in the last
expression was proved above, and, to prove the convergence of
the third integral, with regard to the representation of
$u(x,y;\nu^2)$ in the form (\ref{13,5}),
 it is obviously sufficient to prove that  the series
\begin{equation}
\sum\limits_{k=k_0}^{+\infty}
\int\!\!\!\!\int\limits_{R_k}
 \int\limits_{\lambda_{*}}^{\lambda_{**}}\left|
\int\limits_{P(x,y;\mu)}^{Q(x,y;\mu)} \f''_{\mu\mu}(\widetilde x,\mu)
d\widetilde x\right|^2 d\mu \,dx\,dy
\end{equation}
converges.
The proof of the convergence of this series
repeats verbatim  the proof of the convergence of the series (\ref{sum}) with
the only difference that the bound
 (\ref{ocphi}) is used for  $n=2$.

The cases $n=3,\,4,...$ are treated in a similar way. This completes the proof of
Theorem \ref{teo:trrri}.
 \end{proof}

\section{Distribution of the energy of the initial state of the fluid}
It can readily be proved that   law of conservation of energy holds
 for  the solutions of  problem
(\ref{1})--(\ref{3}) (see, for example, \cite{Zel70}):
$$
{\cal E}(t,D):=\int\limits_D \left(\left|p_{y}\right|^2+
\left|p_{xt}\right|^2+\left|p_{yt}\right|^2\right)  dxdy={\rm const}.
$$
Problem (\ref{1})--(\ref{3}) was studied in \cite{Skaz85}  in the complement
$\R^2\setminus\overline{\Omega}$ to some convex bounded domain $\Omega$.
In this case, the energy of the initial perturbation
is redistributed as $t\to\infty$ in such a way that
the part of energy concentrated on  every compact set
${\cal A}\in\R^2\setminus\overline{\Omega}$
tends to zero, i.e. a scattering of energy occurs.
 In our case the following assertion holds.

 \begin{teo}\label{teo:chetyrrr}
Let $p_i,\,i=1,2$ satisfy the conditions of Theorem
$\ref{teo:trrri}$. In this case, for any
$\varepsilon>0$ and $\delta>0$, there is a
$T=T(\varepsilon,\delta)$ such that
\begin{equation}\label{enereps}
{\cal E}(t,D_{\varepsilon}):=\int
\!\!\!\int\limits_{D_{\varepsilon} }\left(
\left|p_{y}\right|^2+
\left|p_{xt}\right|^2
+\left|p_{yt}\right|^2\right)  \,dx\,dy<\delta
\end{equation}
for the corresponding
solution $p=p(x,y;t)$  of problem {\rm (\ref{1})--(\ref{3})}
for all $t>T$,
where $D_{\varepsilon}$
stands for the set $D\cap\{x>\varepsilon\}\cap\{y<1-\varepsilon\}$.
 \end{teo}

 \begin{proof}
As in the proof of Theorem \ref{teo:trrri},
assume for simplicity that
$\te^{(1)}_1\equiv 0$,
$\te^{(0)}_2\equiv 0,\,\te^{(1)}_2\equiv 0$,
 $\sigma_1(\lambda)\equiv 0$, and
$\supp\sigma_0\subset [\lambda_*,\,\lambda_{**}]
\subset (0,(1+\alpha^2)^{-1})$. In this case
 the solution of problem
(\ref{1})--(\ref{3}) is the function
$$
p(x,y;t)=
\int\limits_{0}^{1}   \cos(\sqrt{\lambda}\; t)\,
\sigma_0(\lambda)\,u(x,y;\lambda)\,d\lambda
=
$$
$$
=-\frac 1t \int\limits_{0}^{1}   \sin(\sqrt{\lambda}\; t)\,
\left(2\sqrt{\lambda} \;\sigma_0(\lambda)
\right)'_{\lambda}
\,u(x,y;\lambda)\,d\lambda
-\frac 1t \int\limits_{0}^{1}   \sin(\sqrt{\lambda}\; t)\,
\sigma_0(\lambda)\,u'_{\lambda}(x,y;\lambda)\;2
\sqrt{\lambda}\,d\lambda
$$
(see (\ref{30})).
As was noted above, under the assumptions of the theorem,
$u(x,y;\lambda)$ is infinitely differentiable with respect to
 $x,y,\lambda$ everywhere inside the prisme
$D\times (0,(1+\alpha^2)^{-1})$ and can be continuously  extended
 together with all its derivatives to
the entire surface of the prism
$$
\left(D\cap\{x>\varepsilon\}\right)\times (0,(1+\alpha^2)^{-1}),
\quad \varepsilon>0.
$$
We have
\begin{eqnarray*}
p_{xt}&=&
-\frac 1t \int\limits_{0}^{1}
 \sqrt{\lambda} \; \sin(\sqrt{\lambda}\; t)\,
\left(2\sqrt{\lambda} \;\sigma_0(\lambda)\right)'_{\lambda}
\,u'_x(x,y;\lambda) \,d\lambda-\\
&&
-\frac 1t \int\limits_{0}^{1}   \sin(\sqrt{\lambda}\; t)\,
\sigma_0(\lambda)\,u''_{x\lambda}(x,y;\lambda)\;2
{\lambda}\,d\lambda,
\end{eqnarray*}
and therefore
$$
\int\!\!\!\!\!\!\!\!\int\limits_{D\cap\{x>\varepsilon\} }
\left|p_{xt}\right|^2  \,dx\,dy
\leq\frac 2{t^2}
\int\!\!\!\!\!\!\!\!\int\limits_{D\cap\{x>\varepsilon\} }
\left| \int\limits_{0}^{1}
 \sqrt{\lambda}  \;\sin(\sqrt{\lambda}\; t)\,
\left(2\sqrt{\lambda} \;\sigma_0(\lambda)
\right)'_{\lambda}
\,u'_x(x,y;\lambda)\, d\lambda
\right|^2 \,dx\,dy+
$$
$$
+\frac 2{t^2}
\int\!\!\!\!\!\!\!\!\int\limits_{D\cap\{x>\varepsilon\} }
\left|
\int\limits_{0}^{1}   \sin(\sqrt{\lambda}\; t)\,
\sigma_0(\lambda)\,u''_{x\lambda}(x,y;\lambda)\;2
{\lambda}\,d\lambda
\right|^2 \,dx\,dy\leq
$$
$$
\leq\frac 2{t^2}
\int\!\!\!\!\!\!\!\!\int\limits_{D\cap\{x>\varepsilon\} }
\int\limits_{0}^{1}
\left|  \sqrt{\lambda}\;
\left(2\sqrt{\lambda}\; \sigma_0(\lambda)
\right)'_{\lambda}
\,u'_x(x,y;\lambda)\right|^2 d\lambda\,\,
 dx\,dy+
$$
$$
+\frac 2{t^2}
\int\!\!\!\!\!\!\!\!\int\limits_{D\cap\{x>\varepsilon\} }
 \int\limits_{0}^{1} \left|
\sigma_0(\lambda)\,u''_{x\lambda}(x,y;\lambda)\;2
{\lambda}\right|^2 d\lambda\,\,
dx\,dy\leq\frac {K_{\varepsilon}}{t^2}\to 0 \quad (t\to \infty),
$$
where the positive constant  $K_{\varepsilon}$ does not depend on $t$.
The other summands in (\ref{enereps}) can be estimated in a similar
way.
 \end{proof}

It follows from the last theorem that, in the course of time,   the total
energy of the initial state of the fluid turns out to be
 almost completely concentrated
in arbitrary small neighborhoods of the vertices
 $O$ and  $B$ of the domain $D$. It is clear here that, if we have
$\te_2^{(i)}\equiv 0$, $i=0,1$, or
${\rm supp}\,\sigma_i\in (0,(1+\alpha^2)^{-1})$, $i=0,1$,
then the energy is accumulated  in a neighborhood of
the point $O$ only. Respectively, if
$\te_1^{(i)}\equiv 0,$ $i=0,1,$ or
${\rm supp}\,\sigma_i\in ((1+\alpha^2)^{-1},1)$, $i=0,1$,
then the energy is accumulated in a neighborhood of
the point  $B$. It is clear that
this picture occurs due to  the fact that
the Poincar\'e-Sobolev equation describes the behavior of
an ideal fluid, whereas,  in the case of a real fluid, one should
consider the corresponding nonlinear systems of equations.

\medskip
In conclusion we note that the approach to the
construction of exact solutions of  problem (\ref{1})--(\ref{3})
which is suggested in the present paper
can be used for a rather wide class of domains with
angular points. For example, let $D$ be
a ``curvilinear triangle'' whose sides
$OA$ and $OB$ are some smooth curves
intersecting at the point $O$ and forming
 a nonzero angle at this point. If, for any
$\lambda \in (\lambda',\lambda'')$ the rays of characteristic
directions whose reflection low at the boundary is described in Section \ref{sec:pervaja}
hide into the angle with  the vertex $O$, then, on this
interval $(\lambda',\lambda'')$,
one can construct differential solutions of the spectral
equation for the operator $A$
 similarly to the rule used in Section \ref{sec:pervaja}, and  the
solutions of the nonstationary problem (\ref{1})--(\ref{3}) corresponding to these
differential solutions. It is clear that
the behavior of these  solutions   as $t\to\infty$
is similar to the behavior  of the solutions  described above.

\end{document}